\newif  \iflong 
\documentclass{llncs}
\usepackage{macros}
\usepackage{multirow}
\usepackage{float}
\usepackage{placeins}
\usepackage{capt-of}
\usepackage[colorlinks]{hyperref}

\hypersetup{
    colorlinks,
    linkcolor={red!85!black},
    citecolor={green!55!black},
    urlcolor={blue!80!black}
}

\begin{document}
  \title{Compact Post-Quantum Signatures from \\ Proofs of Knowledge leveraging Structure \\ for the $\PKP$, $\SD$ and $\RSD$ Problems} 

  \author{Loïc Bidoux$^1$, Philippe Gaborit$^2$}
  \institute{$^1$ Technology Innovation Institute, UAE \\ $^2$ University of Limoges, France}
  \date{}

\maketitle

\begin{abstract}
  The MPC-in-the-head introduced in \cite{IKOS07} has established itself as an important paradigm to design efficient digital signatures.
  For instance, it has been leveraged in the Picnic scheme \cite{picnic} that reached the third round of the NIST Post-Quantum Cryptography Standardization process.
  In addition, it has been used in \cite{Beullens20} to introduce the Proof of Knowledge (PoK) with Helper paradigm.
  This construction permits to design shorter signatures but induces a non negligible performance overhead as it uses cut-and-choose.
  In this paper, we introduce the \emph{PoK leveraging structure} paradigm along with its associated \emph{challenge space amplification} technique.
  Our new approach to design PoK brings some improvements over the PoK with Helper one.
  Indeed, we show how one can substitute the Helper in these constructions by leveraging the underlying structure of the considered problem.
  This new approach does not suffer from the performance overhead inherent to the PoK with Helper paradigm hence offers different trade-offs between security, signature sizes and performances.
  In addition, we also present four new post-quantum signature schemes.
  The first one is based on a new PoK with Helper for the Syndrome Decoding problem.
  It relies on ideas from \cite{BGKM22} and \cite{FJR21} and improve the latter using a new technique that can be seen as performing some cut-and-choose with a meet in the middle approach.
  The three other signatures are based on our new PoK leveraging structure approach and as such illustrate its versatility.
  Indeed, we provide new PoK related to the Permuted Kernel Problem ($\PKP$), Syndrome Decoding ($\SD$) problem and Rank Syndrome Decoding $(\RSD)$ problem.
  In practice, these PoK lead to comparable or shorter signatures than existing ones.
  Indeed, considering (public key + signature), we get sizes below $9$~kB for our signature related to the $\PKP$ problem, below $15$~kB for our signature related to the $\SD$ problem and below $7$~kB for our signature related to the $\RSD$ problem.
  These new constructions are particularly interesting presently as the NIST has recently announced its plan to reopen the signature track of its Post-Quantum Cryptography Standardization process.
\end{abstract}

  %\keywords{PoK leveraging structure \and Signature \and $\PKP$ \and $\SD$ \and $\RSD$}

\newpage

%--------------------------------------------------------------------%
\section{Introduction}
%--------------------------------------------------------------------%

Zero-Knowledge Proofs of Knowledge (PoK) are significant cryptographic primitives thanks to their various applications.
They allow a prover to convince a verifier that he knows some secret without revealing anything about it.
Designing compact PoK is an important problem as one can convert such proofs into signature schemes using the Fiat-Shamir transform \cite{FS, DPJS96} or the Unruh \cite{Unruh15} transform.
Over the years, many post-quantum signatures have been constructed following this approach using for instance the Syndrome Decoding problem \cite{Stern93}, the Multivariate Quadratic problem \cite{Sakumoto11} or the Permuted Kernel Problem \cite{PKP}.
The Picnic \cite{picnic} and MQDSS \cite{MQDSS} signature schemes that have been submitted to the NIST Post-Quantum Cryptography Standardization process also follow this approach.
More recently, Katz, Kolesnikov and Wang \cite{KKW18} proposed a protocol based on the MPC-in-the-head paradigm \cite{IKOS07} in the preprocessing model.
Two years later, Beullens generalized their work by introducing the concept of PoK with Helper \cite{Beullens20}.
The Helper is a trusted third party that can ease the design of PoK and later be removed using cut-and-choose.
Since then, the PoK with Helper paradigm have been extensively used to design post-quantum signature schemes; see for instance \cite{GPS21, FJR21, BGKM22, Wang22}.
This approach is quite interesting as it produces shorter signatures than existing ones however using a Helper along with cut-and-choose also induces a non negligible performance overhead.

In this paper, we introduce the notion of PoK leveraging structure as a new paradigm to design PoK.
Formally, our approach consists in using standard (to be read as without Helper) PoK however we believe that it is more easily understood when described by analogy to the PoK with Helper paradigm.
Indeed, our new approach can be seen as a way to remove the trusted Helper without using cut-and-choose.
In order to do so, we leverage some structure within the hard problem used to design the PoK.
Interestingly, the required additional structure has generally either been already well studied in the literature or is closely related to the considered problem.
PoK following our new framework differs from PoK with Helper ones in several ways.
These differences motivate the introduction of a new technique called challenge space amplification that is particularly well suited for our new way to design PoK.
In practice, PoK leveraging structure can lead to smaller signature schemes than PoK with Helper ones but rely on the use of some structure within the considered hard problem.

\vspace{\baselineskip}
\noindent \textbf{Contributions.}
We propose a new approach to design PoK as well as four new post-quantum signature schemes.
Our main contribution is the introduction of the \emph{Proof of Knowledge leveraging structure} paradigm along with its associated \emph{challenge space amplification} technique.
In addition, we present a new PoK with Helper for the $\SD$ problem that outperforms all existing constructions with the notable exception of the \cite{FJR22} one.
This new PoK is particularly interesting when combined with our PoK leveraging structure framework.
Moreover, we demonstrate the versatility of our new approach by designing three post-quantum signature schemes respectively related to the Permuted Kernel Problem ($\PKP$), Syndrome Decoding ($\SD$) problem and Rank Syndrome Decoding ($\RSD$) problem. 
In practice, these PoK lead to comparable or shorter signatures than existing ones for these problems.
Indeed, considering (public key + signature), we get sizes below $9$~kB for our signature related to the $\PKP$ problem, below $15$~kB for our signature related to the $\SD$ problem and below $7$~kB for our signature related to the $\RSD$ problem.
These new constructions are particularly interesting presently as the NIST has recently announced its plan to reopen the signature track of its Post-Quantum Cryptography Standardization process.

\vspace{\baselineskip}
\noindent \textbf{Paper organization.}
We provide definitions related to PoK, coding theory and hard problems in Section~\ref{sec:preliminaries}.
Our PoK leveraging structure paradigm and our amplification technique are respectively described in Sections~\ref{sec:paradigm} and~\ref{sec:amplification}.
Our PoK with Helper for the $\SD$ problem is depicted in Section~\ref{sec:pok-sd1} while our PoK leveraging structure for the $\PKP$, $\SD$ and $\RSD$ problems are presented in Section~\ref{sec:pokls}.
To finish, we provide a comparison of resulting signature schemes with respect to existing ones in Section~\ref{sec:comparison}.

%--------------------------------------------------------------------%
\section{Preliminaries} \label{sec:preliminaries}
%--------------------------------------------------------------------%

\vspace{0.5\baselineskip}
\noindent \textbf{Notations.} Vectors (respectively matrices) are represented using bold lower-case (respectively upper-case) letters.
%Also, the vectors are assumed to be row vectors by default, and we denote the column vectors by transpose of row of vector (such as $\bm{x}^T$).
%The Hamming weight (number of non-zero coordinates) of a vector $\bm{x}$ is denoted by $\hw{\bm{x}}$.
For an integer $n > 0$, we use $\hperm{n}$ to denote the symmetric group of all permutations of $n$ elements.
In addition, we use $\GLFq{n}$ to denote the linear group of invertible $n \times n$ matrices in $\Fq$.
For a finite set $S$, $x \sampler S$ denotes that $x$ is sampled uniformly at random from $S$ while $x \samples{\theta}$ denotes that $x$ is sampled uniformly at random from $S$ using the seed $\theta$.
Moreover, the acronym $\ppt$ is used as an abbreviation for the term ``probabilistic polynomial time''.
A function is called negligible if for all sufficiently large $\lambda \in \mathbb{N}$, $\negl(\lambda) < \lambda^{-c}$, for all constants $c > 0$.

\subsection{Proof of Knowledge and Commitment schemes}

We start by defining Proof of Knowledge (PoK) following \cite{EPRINT:AttFehKlo21} notations.
Let $R \subseteq \bit^* \times \bit^*$ be an $\textsf{NP}$ relation, we call $(x, w) \in R$ a statement-witness pair where $x$ is the statement and $w$ is the witness.
The set of valid witnesses for $x$ is $R(x) = \{ w \,|\, (x, w) \in R \}$.
In a PoK, given a statement $x$, a prover $\prover$ aims to convice a verifier $\verifier$ that he knows a witness $w \in R(x)$.

\begin{definition}[Proof of Knowledge]
  A $(2n + 1)$-round PoK for relation R with soundness error $\epsilon$ is a two-party protocol between a prover $\prover(x, w)$ with input a statement $x$ and witness $w$ and a verifier $\verifier(x)$ with input $x$.
We denote by $\big \langle \prover(x, w), \allowbreak \verifier(x) \big \rangle$ the transcript between $\prover$ and $\verifier$. 
A PoK is correct if
\vspace{-0.25\baselineskip}
\begin{equation*}
  \prb
  \left[ \begin{array}{l}
    \accept \leftarrow \big \langle \prover(x, w), \verifier(x) \big \rangle
  \end{array} \right]
  = 1.
\end{equation*}
\end{definition}

\begin{definition} [Tree of Transcripts] \label{def:tree-of-transcripts}
  Let $k_1, \ldots, k_n \in \mathbb{N}$, a $(k_1, \ldots, \allowbreak k_n)$-tree of transcripts for a $(2n+1)$-round public coin protocol
  PoK is a set of $K = \prod_{i = 1}^{n} k_i$ transcripts arranged in a tree structure. The nodes in the tree
  represent the prover's messages and the edges between the nodes correspond to the challenges sent by the verifier. Each node at depth
  $i$ has exactly $k_i$ children corresponding to the $k_i$ pairwise distinct challenges. Every transcript is represented by exactly
  one path from the root of the tree to a leaf node. 
\end{definition}

\begin{definition} [$(k_1, \ldots, k_n)$-out-of-$(N_1, \ldots, N_n)$ Special-Soundness] \label{def:k-soundness}
  Let $k_1, \allowbreak \ldots, k_n, N_1, \ldots, N_n \in \mathbb{N}$. A $(2n + 1)$-round public-coin $PoK$,
  where $\verifier$ samples the $i$-th challenge from a set of cardinality $N_i \geq k_i$ for $i \in [n]$, is
  $(k_1, \ldots, k_n)$-out-of-$(N_1, \ldots, N_n)$ special-sound if there exists a $\ppt$ algorithm that
  on an input statement $x$ and a $(k_1, \ldots, k_n)$-tree of accepting transcripts outputs a witness $w$. We also say
  $PoK$ is $(k_1, \ldots, k_n)$-special-sound.
\end{definition}

\begin{definition}[Special Honest-Verifier Zero-Knowledge]
  A PoK satisfies the Honest-Verifier Zero-Knowledge (HZVK)  property if there exists a $\ppt$ simulator $\simulator$ that given as input a statement $x$ and random challenges $(\kappa_1, \ldots, \kappa_n)$, outputs a transcript $\big \langle \simulator(x, \kappa_1, \cdots, \kappa_n), \allowbreak \verifier(x) \big \rangle$ that is computationally indistinguishable from the probability distribution of transcripts of honest executions between a prover $\prover(x, w)$ and a verifier $\verifier(x)$.
\end{definition}

PoK with Helper introduced in \cite{Beullens20} are protocols that leverage a trusted third party (the so-called Helper) within their design.
They can be seen as 3-round PoK following an initial step performed by the Helper.
We defer the reader to \cite{Beullens20} for their formal definition.
PoK are significant cryptographic primitives as they can be turned into digital signatures using the Fiat-Shamir transform \cite{FS, DPJS96, EPRINT:AttFehKlo21}.
Hereafter, we define commitment schemes which are an important building block used to construct PoK.

\begin{definition}[Commitment Scheme]
A commitment scheme is a tuple of algorithms $(\mathsf{Com}, \mathsf{Open})$ such that $\mathsf{Com}(r, m)$ returns a commitment $c$ for the message $m$ and randomness $r$ while $\mathsf{Open}(c, r, m)$ returns either $1$ ($\mathsf{accept}$) or $0$ ($\mathsf{reject}$). 
A commitment scheme is said to be correct if:
\vspace{-0.25\baselineskip}
 \begin{equation*}
 \prb
  \left[ b = 1 \ \middle\vert \begin{array}{l}
    c \leftarrow \commit{r, m}, ~ b \leftarrow \open{c, r, m}
  \end{array}\right]
  =1.
  \end{equation*}
\end{definition}

\begin{definition}[Computationally Hiding]
Let $(m_0, m_1)$ be a pair of messages, the advantage of $\adv$ against the hiding experiment is defined as:
\begin{equation*}
  \mathsf{Adv}^{\mathsf{hiding}}_{\adv}(1^\lambda) = 
  \Bigg| \, \prb \left[ \begin{array}{l}
    b = b' \\
  \end{array} \ \middle\vert~ \\ \begin{array}{l}
    b \sampler \{0, 1\}, ~ r \sampler \bit^\lambda \\
    c \samplen \commit{r, m_b}, ~ b' \samplen \adv.\mathsf{guess}(c)
  \end{array}\right]
  - \frac{1}{2} \, \Bigg|.
\end{equation*}
A commitment scheme is computationally hiding if for all $\ppt$ adversaries $\adv$ and every pair of messages $(m_0, m_1)$, $\mathsf{Adv}^{\mathsf{hiding}}_{\adv}(1^\lambda)$ is negligible in $\lambda$.
\end{definition}

\begin{definition}[Computationally Binding]
The advantage of an adversary $\adv$ against the commitment binding experiment is defined as:
\begin{equation*}
  \mathsf{Adv}^{\mathsf{binding}}_{\adv}(1^\lambda) = \prb
  \left[ \begin{array}{l}
    m_0 \neq m_1 \\
    1 \samplen \open{c, r, m_0} \\ 
    1 \samplen \open{c, r, m_1} \\ 
  \end{array} \ \middle\vert~ \\ \begin{array}{l}
    (c, r, m_0, m_1) \samplen \adv.\mathsf{choose}(1^{\lambda})
  \end{array}\right].
\end{equation*}
A commitment scheme is computationally binding if for all $\ppt$ adversaries $\adv$, $\mathsf{Adv}^{\mathsf{binding}}_{\adv}(1^\lambda)$ is negligible in $\lambda$.
\end{definition}

\subsection{Coding theory}

We recall some definitions for both Hamming and rank metrics.
Let $n$ be a positive integer, $q$ a prime power, $m$ a positive integer, $\Fqm$ an extension of degree $m$ of $\Fq$ and $\beta := (\beta_1, \dots, \beta_m)$ a basis of $\Fqm$ over $\Fq$. 
Any vector $\bm{x} \in \Fqmn$ can be associated to the matrix $\bm{M_x} \in \Fq^{m \times n}$ by expressing its coordinates in $\beta$.

\begin{definition}[Hamming weight]
Let $\bm{x} \in \Ftn$, the Hamming weight of $\bm{x}$, denoted $\hw{\bm{x}}$, is the number of non-zero coordinates of $\bm{x}$.
\end{definition}

\begin{definition}[Rank weight]
Let $\bm{x} \in \Fqmn$, the rank weight of $\bm{x}$, denoted $\rw{\bm{x}}$, is defined as the rank of the matrix $\bm{M_x} = (x_{ij}) \in \Fq^{m \times n}$ where $x_j = \sum^m_{i = 1} x_{i,j} \beta_i$. 
\end{definition}

\begin{definition}[Support]
Let $\bm{x} \in \Fqmn$, the support of $\bm{x}$ denoted $\rsup{\bm{x}}$, is the $\Fq$-linear space generated by the coordinates of $\bm{x}$ namely $\rsup{\bm{x}} = \langle x_1, \dots, x_n \rangle_{\Fq}$.
It follows from the definition that $\rw{\bm{x}} = | \rsup{\bm{x}} |$.
\end{definition}

We define linear codes over a finite field $\mathbb{F}$ where $\mathbb{F} = \Ft$ in Hamming metric and $\mathbb{F} = \Fqm$ in rank metric as well as quasi-cyclic codes and ideal codes.
We restrict our definitions to codes of index 2 as they are the ones used hereafter.

\begin{definition}[$\mathbb{F}$-linear code]
An $\mathbb{F}$-linear code $\mathcal{C}$ of length $n$ and dimension $k$ denoted $[n, k]$ is an $\mathbb{F}$-linear subspace of $\mathbb{F}^n$ of dimension $k$.
A generator matrix for $\mathcal{C}$ is a matrix $\bm{G} \in \mathbb{F}^{k \times n}$ such that $\mathcal{C} = \left\lbrace \bm{x} \bm{G}, ~ \bm{m} \in \mathbb{F}^k \right\rbrace$.
  A parity-check matrix for $\mathcal{C}$ is a matrix $\bm{H} \in \mathbb{F}^{(n-k) \times n}$ such that $\mathcal{C} = \left\lbrace \bm{x} \in \mathbb{F}^n, ~ \bm{H} \bm{x}^{\top} = 0 \right\rbrace$.
\end{definition}

\begin{definition}[Quasi-cyclic code]
A systematic binary quasi-cyclic code of index 2 is a $[n = 2k, k]$ code that can be represented by a generator matrix $\bm{G} \in \Ft^{k \times n}$ of the form $\mathbf{G} = [ \mathbf{I}_k ~ \mathbf{A}]$ where $\mathbf{A}$ is a circulant $k\times k$ matrix. 
Alternatively, it can be represented by a parity-check matrix $\bm{H} \in \Ft^{(n - k) \times k}$ of the form $\mathbf{H} = [\mathbf{I}_k ~ \mathbf{B}]$ where $\mathbf{B}$ is a circulant $k \times k$ matrix.
\end{definition}

\begin{definition}[Ideal matrix]
Let $P \in \Fq[X]$ a polynomial of degree $k$ and let $\bm{v} \in \Fqm^k$.
The ideal matrix $\mathcal{IM}_P(\bm{v}) \in \Fqm^{k \times k}$ is of the form
\begin{equation*}
	\mathcal{IM}_{P}(\mathbf{v}) = \begin{pmatrix}
	 	\hspace{32pt} \bm{v} \mod P \\
	 	\hspace{16pt} X \cdot \bm{v} \mod P \\
	 	\vdots \\
	 	X^{k-1} \cdot \bm{v} \mod P
	 \end{pmatrix}.
\end{equation*}
\end{definition}

\begin{definition}[Ideal code]
Let $P \in \Fq[X]$ a polynomial of degree $k$, $\bm{g} \in \Fqm^k$ and $\bm{h} \in \Fqm^k$.
An ideal code of index 2 is a $[n = 2k, k]$ code $\mathcal{C}$ that can be represented by a generator matrix $\bm{G} \in \Fqm^{k \times n}$ of the form $\mathbf{G} = [ \mathbf{I}_k ~ \mathcal{IM}_{P}(\bm{g})]$.
Alternatively, it can be represented by a parity-check matrix $\bm{H} \in \Fqm^{(n - k) \times k}$ of the form $\mathbf{H} = [\mathbf{I}_k ~ \mathcal{IM}_{P}(\bm{h})]$.
\end{definition}

Let $\bm{a} = (a_1, \ldots, a_{k}) \in \Ftk$, for $r \in \intoneto{k - 1}$, we define the $\rot()$ operator as $\rot_r(\bm{a}) = (a_{k - r + 1}, \ldots, a_{k - r})$.
For $\bm{b} = (\bm{b}_1, \bm{b}_2) \in \Ft^{2k}$, we slightly abuse notations and define $\rot_r(\bm{b}) = (\rot_r(\bm{b}_1), \rot_r(\bm{b}_2))$.
Whenever $\bm{H}$ is the parity-check matrix of a quasi-cyclic code or an ideal code, if one has $\bm{H} \bm{x}^{\top} = \bm{y}^{\top}$, then it holds that $\bm{H} \cdot \rot_r(\bm{x})^{\top} = \rot_r(\bm{y})^{\top}$.

\subsection{Hard problems}

We introduce several hard problems along with some of their variants.
They are used to design signature schemes in the remaining of this paper.

\vspace{0.25\baselineskip}
\begin{definition}[$\PKP$ problem] \label{def:pkp}
  Let $(q, m, n)$ be positive integers, $\bm{H} \in \Fq^{m \times n}$ be a random matrix, $\pi \in \hperm{n}$ be a random permutation and $\bm{x} \in \Fq^{n}$ be a vector such that $\bm{H}(\pi[\bm{x}]) = 0$.
  Given $(\bm{H}, \bm{x})$, the Permuted Kernel Problem $\PKP(q, m, \allowbreak n)$ asks to find a permutation $\pi$.
\end{definition}

\vspace{0.25\baselineskip}
\begin{definition}[$\IPKP$ problem] \label{def:ipkp}
  Let $(q, m, n)$ be positive integers, $\bm{H} \sampler \Fq^{m \times n}$ be a random matrix, $\pi \in \hperm{n}$ be a random permutation, $\bm{x} \in \Fq^{n}$ be a random vector and $\bm{y} \in \Fq^{m}$ be a vector such that $\bm{H}(\pi[\bm{x}]) = \bm{y}$.
  Given $(\bm{H}, \bm{x}, \bm{y})$, the Inhomogeneous Permuted Kernel Problem $\IPKP(q, m, n)$ asks to find $\pi$.
\end{definition}

\vspace{0.25\baselineskip}
\begin{definition}[$\SD$ problem] \label{def:sd}
  Let $(n, k, w)$ be positive integers, $\bm{H} \in \Ft^{\nmktn}$ be a random parity-check matrix, $\bm{x} \in \Ft^n$ be a random vector such that $\hw{\bm{x}} = w$ and $\bm{y} \in \Ft^{(\nmk)}$ be a vector such that $\bm{Hx}^\top = \bm{y}^\top$.
  Given $(\bm{H}, \bm{y})$, the binary Syndrome Decoding problem $\SD(n, k, w)$ asks to find $\bm{x}$.
\end{definition}

\vspace{0.25\baselineskip}
\begin{definition}[$\QCSD$ problem] \label{def:qcsd}
  Let $(n=2k, k, w)$ be positive integers, $\bm{H} \in \mathcal{QC}(\Ft^{(n - k) \times n})$ be a random parity-check matrix of a quasi-cyclic code of index~$2$, $\bm{x} \in \Ft^{n}$ be a random vector such that $\hw{\bm{x}} = w$ and $\bm{y} \in \Ft^{(n-k)}$ be a vector such that $\bm{Hx}^\top = \bm{y}^\top$.
  Given $(\bm{H}, \bm{y})$, the binary Quasi-Cyclic Syndrome Decoding problem $\QCSD(n, k, w)$ asks to find $\bm{x}$.
\end{definition}

\vspace{0.25\baselineskip}
\noindent \textbf{Decoding One Out of Many setting.} We denote by $\QCSD(n, k, w, M)$ the $\QCSD$ problem in the decoding one out of many setting \cite{Sen11}.
In this setting, several small weight vectors $(\bm{x}_i)_{i \in \intoneto{M}}$ are used along with several syndromes $(\bm{y}_i)_{i \in \intoneto{M}}$ and one is asked to find any $\bm{x}_i$.
This setting is natural for the $\QCSD$ problem as one can get additional syndromes using rotations.

%\vspace{0.25\baselineskip}
%\textcolor{red}{
%\begin{definition}[$\DOOMQCSD$ problem]
%  Let $(q, n=\ell k, k, w, M)$ be positive integers, $\bm{H} \in \mathcal{QC}(\Fq^{(n - k) \times n})$ be a random parity-check matrix of a quasi-cyclic code of index $\ell$, $(\bm{x}_1, \cdots, \bm{x}_M) \in (\Fq^{n})^M$ be random vectors such that $\hw{\bm{x}_i} = w$ foreach $i \in \intoneto{M}$ and $(\bm{y}_1, \cdots, \bm{y}_M) \in (\Fq^{(n-k)})^M$ be vectors such that $\bm{H} \bm{x}_i^\top = \bm{y}_i^\top$ foreach $i \in \intoneto{M}$.
%  Given $(\bm{H}, \bm{y}_1, \cdots, \bm{y}_M)$, the Decoding One Out of Many Quasi-Cyclic Syndrome Decoding problem $\DOOMQCSD(q, n, k, w, M)$ asks to find $\bm{x}_i$ for any $i \in \intoneto{M}$.
%\end{definition}}

\vspace{0.5\baselineskip}
\begin{definition}[$\RSD$ problem] \label{def:rsd}
  Let $(q, m, n, k, w)$ be positive integers, $\bm{H} \in \Fqm^{\nmktn}$ be a random parity-check matrix, $\bm{x} \in \Fqm^n$ be a random vector such that $\rw{\bm{x}} = w$ and $\bm{y} \in \Fqm^{(\nmk)}$ be a vector such that $\bm{Hx}^\top = \bm{y}^\top$.
  Given $(\bm{H}, \bm{y})$, the Rank Syndrome Decoding problem $\RSD(q, m, n, k, w)$ asks to find~$\bm{x}$.
\end{definition}

\vspace{0.25\baselineskip}
\begin{definition}[$\IRSD$ problem] \label{def:irsd}
  Let $(q, m, n = 2k, k, w)$ be positive integers, $P \in \Fq[X]$ be an irreducible polynomial of degree $k$, $\bm{H} \in \mathcal{ID}(\Fqm^{(n - k) \times n})$ be a random parity-check matrix of an ideal code of index $2$, $\bm{x} \in \Fqm^{n}$ be a random vector such that $\rw{\bm{x}} = w$ and $\bm{y} \in \Fqm^{(n - k)}$ be a vector such that $\bm{Hx}^\top = \bm{y}^\top$.
  Given $(\bm{H}, \bm{y})$, the Ideal Rank Syndrome Decoding problem $\IRSD(q, m, n, k, w)$ asks to find~$\bm{x}$.
\end{definition}

\vspace{0.25\baselineskip}
\begin{definition}[$\RSL$ problem] \label{def:rsl}
  Let $(q, m, n, k, w, M)$ be positive integers, $\bm{H} \in \Fqm^{\nmktn}$ be a random parity-check matrix, $E$ be a random subspace of $\Fqm$ of dimension $\omega$, $(\bm{x}_i)_{i \in \intoneto{M}} \in (\Fqm^n)^M$ be random vectors such that $\rsup{\bm{x}_i} = E$ and $(\bm{y}_i) \in (\Fqm^{(\nmk)})^M$ be vectors such that $\bm{H} \bm{x}_i^\top = \bm{y}_i^\top$.
  Given $(\bm{H}, (\bm{y}_i)_{i \in \intoneto{M}})$, the Rank Support Learning problem $\RSL(q, m, n, k, w, M)$ asks to find~E.
\end{definition}

\vspace{0.25\baselineskip}
\begin{definition}[$\IRSL$ problem] \label{def:irsl}
  Let $(q, m, n =2k, k, w, M)$ be positive integers, $P \in \Fq[X]$ be an irreducible polynomial of degree $k$, $\bm{H} \in \mathcal{ID}(\Fqm^{\nmktn})$ be a random parity-check matrix of an ideal code of index $2$, $E$ be a random subspace of $\Fqm$ of dimension $\omega$, $(\bm{x}_i)_{i \in \intoneto{M}} \in (\Fqm^n)^M$ be random vectors such that $\rsup{\bm{x}_i} = E$ and $(\bm{y}_i)_{i \in \intoneto{M}} \in (\Fqm^{(\nmk)})^M$ be vectors such that $\bm{H} \bm{x}_i^\top = \bm{y}_i^\top$.
  Given $(\bm{H}, (\bm{y}_i)_{i \in \intoneto{M}})$, the Ideal Rank Support Learning problem $\IRSL(q, m, n, k, \allowbreak w, M)$ asks to find~E.
\end{definition}

\section{New paradigm to design PoK leveraging structure} \label{sec:paradigm}
%--------------------------------------------------------------------%

\subsection{Overview of our PoK leveraging structure approach}

The PoK with Helper paradigm introduced in \cite{Beullens20} eases the design of PoK and has historically led to shorter signatures.
It relies on introducing a trusted third party (the so called Helper) that is later removed using cut-and-choose.
In this section, we introduce a new paradigm that can be seen as an alternate way to remove the Helper in such PoK.
Formally, our approach consists in using standard (to be read as without Helper) PoK however we believe that it is more easily understood when described by analogy to the PoK with Helper paradigm.

\vspace{\baselineskip}
\noindent \textbf{PoK leveraging structure.} Our main idea can be seen as using the underlying structure of the considered problem in order to substitute the Helper in the PoK with Helper approach.
As a consequence, all the PoK designed following our paradigm share the framework described in Figure~\ref{fig:paradigm}.
Indeed, all our PoK leveraging structure are 5-round PoK whose first challenge space $\mathcal{C}_{\mathsf{struct}}$ is related to some structure and whose second challenge space can be made arbitrarily large.
This specific framework impacts the properties of the considered protocols which can be leveraged to bring an improvement over existing constructions.
Indeed, removing the Helper provides an improvement on performances that can lead to smaller signature sizes when some additional conditions are verified.
We describe in Table~\ref{table:paradigm} how our new paradigm can be applied to the PoK with Helper from \cite{Beullens20} and Section~\ref{sec:pok-sd1} in order to design new PoK leveraging structure related to the $\PKP$, $\SD$ and $\RSD$ problems.
For instance, starting from the $\SD$ problem over $\Ft$, one may introduce the required extra structure by considering the $\QCSD$ problem over $\Ft$ or the $\SD$ problem over $\Fq$.

\begin{figure}[!h]
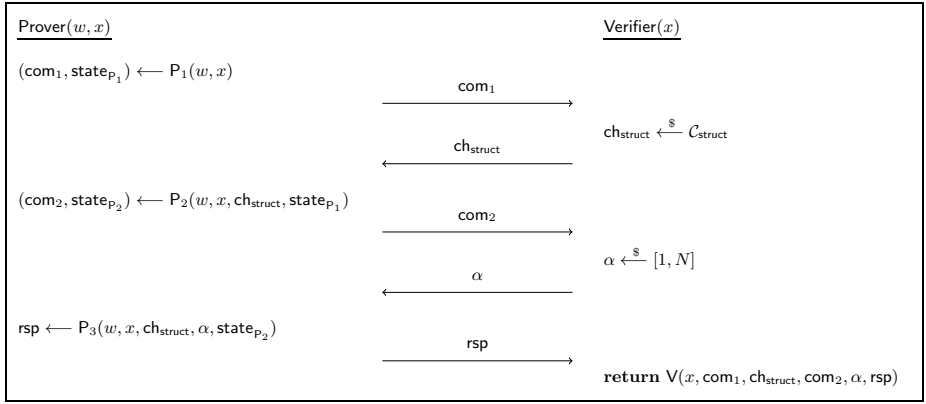
 
  \begin{center}
    \resizebox{1\textwidth}{!}{\fbox{
      \pseudocode{%
        \hspace{160pt} \> \> \hspace{160pt} \\[-0.75\baselineskip] % Hack to force font scaling
        \underline{\proverf(w, x)} \> \> \underline{\verifierf(x)} \\[0.5\baselineskip]
        (\com_1, \statepone) \samplen \prover_1(w, x) \\[-0.75\baselineskip]
        \> \sendmessageright*{\com_1} \\[-0.25\baselineskip]
        \> \> \chs \sampler \chsps\\[-0.75\baselineskip]
        \> \sendmessageleft*{\chs} \\[0.25\baselineskip]
        %2
        (\com_2, \stateptwo) \samplen \prover_2(w, x, \chs, \statepone) \hspace{10pt }\\[-0.75\baselineskip]
        \> \sendmessageright*{\com_2} \\[-0.25\baselineskip]
        \> \> \alpha \sampler \intoneto{N} \\[-0.75\baselineskip]
        \> \sendmessageleft*{\alpha} \\[0.25\baselineskip]
        \rsp \samplen \prover_3(w, x, \chs, \alpha, \stateptwo) \\[-0.75\baselineskip]
        \> \sendmessageright*{\rsp} \\[-0.5\baselineskip]
        \> \> \pcreturn \verifier(x, \com_1, \chs, \com_2, \alpha, \rsp) \\[-0.75\baselineskip]
      }
    }}
    \caption{Overview of PoK leveraging structure} \label{fig:paradigm}
  \end{center}
\vspace{-\baselineskip}
\end{figure}

\begin{table}[ht]
\begin{center}
{\setlength{\tabcolsep}{0.35em}
{\renewcommand{\arraystretch}{1.6}
{\scriptsize
  \begin{tabular}{|l|l|l|l|}
    \hline
    \multicolumn{1}{|c}{Scheme} & \multicolumn{1}{|c}{First Challenge} & \multicolumn{1}{|c}{Second Challenge} & \multicolumn{1}{|c|}{Problem} \\ \hline \hline

    SUSHYFISH \cite{Beullens20}                              & Helper                 & Linearity over $\Fq$ & $\IPKP$ \\ \hline
    Section \ref{sec:pok-pkp2}                               & Linearity over $\Fq$   & Shared Permutation   & $\IPKP$ \\ \hline \hline

    \cite{FJR21}, \cite{BGKM22}                             & Helper                 & (Shared Permutation) & $\SD$ over $\Ft$ \\ \hline
    Section \ref{sec:pok-sd1}                                & Helper                 & Shared Permutation & $\SD$ over $\Ft$ \\ \hline
    Section \ref{sec:pok-sd2}                                & Cyclicity              & Shared Permutation   & $\QCSD$ over $\Ft$ \\ \hline \hline
%    Appendix \ref{app:pok-sd2-fq}                             & Linearity over $\Fq$   & Shared Permutation   & $\SD$ over $\Fq$ \\ \hline \hline

    Section \ref{sec:pok-sd1}                                & Helper                 & Shared Permutation   & $\RSD$ \\ \hline
    \multirow{2}{*}{Section \ref{sec:pok-rsd2}}              & Cyclicity              & Shared Permutation   & $\IRSD$ \\ \cline{2-4}
                                                             & Cyclicity, Linearity, Same Support & Shared Permutation & $\IRSL$ \\ \hline

%    MUDFISH \cite{Beullens20}                                & Helper                 & Linearity over $\Fq$ & $\MQ$ \\ \hline
%    \cite{Wang22}                                            & Helper                 & Secret Sharing       & $\MQ$ \\ \hline
%    Section \ref{sec:pok-mq2}                                & Homogeneity over $\Fq$ & Secret Sharing       & $\HMQ$ \\ \hline \hline

  \end{tabular}
  \vspace{0.5\baselineskip}
  \caption{Substitution of the Helper using PoK leveraging structure} \label{table:paradigm}
}}}
\end{center}
\end{table}
\vspace{-2\baselineskip}

\subsection{Properties of PoK leveraging structure}

Hereafter, we discuss the advantages and limits of our new approach by des-cribing its impact on soundness, performance, security and size of the resulting signature scheme.
We denote by $\cC_1$ and $\cC_2$ the sizes of the two considered challenge spaces and by $\tau$ the number of iterations that one need to execute to achieve a negligible soundness error for some security parameter $\lambda$.

\vspace{\baselineskip}
\noindent \textbf{Impact on soundness error.}
Removing the trusted setup induced by the Helper allows new cheating strategies for malicious provers.
In most cases, this means that the soundness error is increased from $\max(\frac{1}{\cC_1}, \frac{1}{\cC_2})$ for protocols using a Helper to $\frac{1}{\cC_1} + (1 - \frac{1}{\cC_1}) \cdot \frac{1}{\cC_2}$ for protocols without Helper.
A malicious prover has generally a cheating strategy related to the first challenge that is sucessful with probability $\frac{1}{\cC_1}$ as well as a cheating strategy related to the second challenge for the remaining cases that is sucessful with probability $(1 - \frac{1}{\cC_1}) \cdot \frac{1}{\cC_2}$. 
The aforementioned expression is equal to $\frac{1}{\cC_2} + \frac{\cC_2 - 1}{\cC_1 \cdot \cC_2}$ hence one can see that for challenge spaces such that $\cC_2 \gg 1$, the resulting soundness error is close to $\frac{1}{\cC_2} + \frac{1}{\cC_1}$.
In addition, if one has also $\cC_1 \gg \cC_2$, the soundness error is close to $\frac{1}{\cC_2}$.

\vspace{\baselineskip}
\noindent \textbf{Impact on performances.}
Using a Helper, one has to repeat several operations (some sampling and a commitment in most cases) $\tau \cdot \cC_1 \cdot \cC_2$ times during the trusted setup phase where $\cC_1$ is the number of repetitions involved in the cut and choose used to remove the Helper.
While in practice the \emph{beating parallel repetition} optimization from \cite{KKW18} allows to reduce the number of operations to $X \cdot \cC_2$ where $X \leq \tau \cdot \cC_1$, the trusted setup phase still induces an important performance overhead.
Removing the Helper generally reduces the cost of this phase to $\tau \cdot \cC_2$ operations thus improving performances as $\tau \leq X$. 
One should note that the PoK constructed from our new technique inherently have a 5-round structure.
A common technique consists to collapse 5-round PoK to 3-round however doing so would re-introduce the aforementioned performance overhead in our case.

\vspace{\baselineskip}
\noindent \textbf{Impact on security.}
PoK following our new paradigm are slightly less conservative than PoK with Helper collapsed to 3-round. 
Indeed, the underlying 5-round structure can be exploited as demonstrated by the attack from \cite{KZ20}.
In practice, one can increase the number of iterations $\tau$ to take into account this attack.
One can also amplify the challenge spaces sizes in order to reduce the effectiveness of this attack as explained in Section~\ref{sec:amplification}.
In addition, the security proofs of the PoK following our new paradigm might be a bit more involved.
Indeed, one might need to introduce an intermediary problem and rely on a reduction from the targeted problem to this intermediary problem.
This strategy was used in \cite{AGS11} with the introduction of the $\DiffSD$ problem and its reduction to the $\QCSD$ problem.
More recently, such a strategy has also been used in \cite{FJR22} where the d-split $\SD$ problem is used as an intermediary problem along with a reduction to the $\SD$ problem.
In practice, one generally need to increase the considered parameters for these reductions to hold securely.

\vspace{\baselineskip}
\noindent \textbf{Impact on resulting signature size.}
Several of the aforementioned elements impact the signature size in conflicting ways.
For instance, removing the Helper increases the soundness which impacts the signature size negatively however it also reduces the number of seeds to be sent which impacts the signature size positively.
In addition, many PoK with Helper feature a trade-off between performances and signature sizes hence our performance improvement related to the Helper's removal directly translate into a reduction of the signature size in these cases.
Moreover, using a 5-round structure reduces the number of commitments to be sent but requires to take into account the attack from \cite{KZ20}.
In practice, our new paradigm lead to smaller signature sizes over comparable PoK with Helper when the parameters are chosen carefuly.
This is mainly due to the performance improvement related to the Helper's removal while the work from \cite{KZ20} constitutes a limiting factor thus motivating the new technique introduced in Section~\ref{sec:amplification}.

\vspace{\baselineskip}
Our new paradigm exploits the structure of the considered problems in order to build more compact PoK.
As a result, our protocols differs from existing ones in several ways thus leading to different features when it comes to soundness, performances, security and resulting signature sizes.
Interestingly, the required additional structure has been either already well studied (such as cyclicity for the $\QCSD$ and $\IRSD$ problems) or is closely related to  the considered problem (such as linearity over $\Fq$ for the inhomegeneous $\IPKP$ problem).  
%Interestingly, the required additional structure has been either already well studied (such as cyclicity for the $\QCSD$ and $\IRSD$ problems) or is closely related to  the considered problem (such as linearity over $\Fq$ for the inhomegeneous $\IPKP$ problem or homogeneity over $\Fq$ for the homogeneous $\HMQ$ problem).  

%--------------------------------------------------------------------%
\section{Amplifying challenge space to mitigate \cite{KZ20} attack} \label{sec:amplification}
%--------------------------------------------------------------------%

\subsection{Overview of our challenge space amplification technique}

It was shown in \cite{KZ20} that 5-round PoK that use parallel repetition to achieve a negligible soundness error are vulnerable to an attack when they are made non-interactive with the Fiat-Shamir transform.
One can easily mitigate this attack by increasing the number of considered repetitions $\tau$.
This countermeasure increases the resulting signature size hence most designers choose to collapse 5-round schemes into 3-round ones instead.
In the case of our new paradigm, such a strategy is not desirable as explained in Section~\ref{sec:paradigm}.
These considerations motivate the research for an alternate way to mitigate the \cite{KZ20} attack.

\vspace{\baselineskip}
\noindent \textbf{Challenge space amplification.} Our new mitigation strategy consists in amplifying the challenge space sizes rather than increasing the number of repetitions.
In particular, we are interested in amplifying the first challenge space of our PoK leveraging structure.
Hereafter, we show how such an amplification can be performed by increasing the size of the public key.
Interestingly, in some cases, one may increase the size of the challenge space exponentially while increasing the size of the public key only linearly.
In such cases, our new mitigation strategy is more efficient than increasing the number of repetitions $\tau$.

\subsection{The \cite{KZ20} attack against 5-round PoK}

We start by recalling that most PoK from the litterature feature a soundness error ranging from $2/3$ to $1/N$ for $N$ up to $1024$.
In order to achieve a negligible soundness error, one has to execute $\tau$ repetitions of the underlying protocol. 
The main idea of the \cite{KZ20} result is to separate the attacker work in two steps respectively related to the first and second challenges.
A malicious prover tries to correctly guess the first challenges for $\tau^* \leq \tau$ repetitions as well as the second challenges for the remaining $\tau - \tau^*$ repetitions leveraging the fact that a correct guess for any challenge allows to cheat in most protocols.
The efficiency of the attack depends on the number of repetitions $\tau$, the sizes of the challenge spaces $\cC_1$ and $\cC_2$ as well as a property of the PoK called the capability for early abort.
PoK with Helper have the capability for early abort however schemes constructed following our new paradigm don't.
As a consequence, finding a way to mitigate the \cite{KZ20} attack is of great importance in our context.
Our mitigation strategy relies on the fact that the attack complexity increases if the two challenges spaces $\cC_1$ and $\cC_2$ are not of equal size.
Indeed, in our context (no capability for early abort), the attack complexity is equal to
$ \mathsf{Cost}(\tau) = \frac{1}{P_1(\tau^*, \tau, \cC_1)} + (\cC_2)^{\tau - \tau^*}$
 where $P_1(r, \tau, \cC_1)$ is the probability to correctly guess at least $r$ amongst $\tau$ challenges in a challenge space of size $\cC_1$ namely
$ P_1(r, \tau, \cC_1) = \sum\nolimits_{i = r}^{\tau} \Big( \frac{1}{\cC_1} \Big)^i \cdot \Big( \frac{\cC_1 - 1}{\cC_1} \Big)^{\tau - i} \cdot {\tau \choose i}$
 with $\tau^*$ being the number of repetitions in the first step of the attack that minimizes the attack cost namely $\tau^* = \underset{0 \leq r \leq \tau}{\arg \min} ~  \frac{1}{P_1(r, \tau, \cC_1)} + (\cC_2)^{\tau - r}$.

\subsection{Trade-off between public key size and challenge space size}

In order to mitigate the attack from \cite{KZ20}, we propose to increase the size of the first challenge space.
This decreases the probability that a malicious prover is able to correctly guess $\tau^*$ random challenges in the first step of the attack hence reduces its efficiency.
In order to do so, one can introduce several instances of the considered hard problem in a given public key as suggested in \cite{ISIT21}.
Using the $\SD$ problem for illustrative purposes, one replaces the secret key $\sk = (\bm{x})$ where $\hw{\bm{x}} = \omega$ and public key $\pk = (\bm{H}, \bm{y}^\top = \bm{H} \bm{x}^{\top})$ by $M$ instances such that $\sk = (\bm{x}_i)_{i \in \intoneto{M}}$ and $\pk = (\bm{H}, (\bm{y}^{\top}_i = \bm{H} \bm{x}^{\top}_i)_{i \in \intoneto{M}})$ where $\hw{\bm{x}_i} = \omega$.
In practice, this implies that the parameters of the schemes have to be chosen taking into account that the attacker has now access to $M$ instances $(\bm{H}, (\bm{y}_i)_{i \in \intoneto{M}})$.
This setting has been studied within the code-based cryptography community and is commonly referred as decoding one out of many in the literature \cite{Sen11}.
Using this technique, the first challenge space size is increased by a factor $M$ however the size of the public key is also increased by a factor $M$.
%Overall, this offers an interesting trade-off between public key size and signature size as suggested by parameters 
Overall, the challenge space amplification approach reduces the (public key + signature) size as suggested by numbers in Section~\ref{sec:comparison}. % as mitigating the \cite{KZ20} attack allows to reduce the signature size.
%Indeed, using our $\QCSD$ based PoK for illustrative purposes (see Section~\ref{sec:param-sd} for more details), one can see that amplifying the first challenge space reduces the number of repetitions $\tau$ from $43$ to $37$ namely an improvement approximately equal to~$15\%$.
%For instance, the PoK with Helper from Section~\ref{sec:pok-sd1} features a (public key + signature) size of $19.7$ kB while its PoK leveraging structure counterpart reduces this amount to $17$ kB. 

\vspace{\baselineskip}
We now present a new idea that greatly improves the efficiency of the aforementioned amplification technique.
Interestingly, in some cases, it is possible to increase the challenge space size \emph{exponentially} while only increasing the public key size \emph{linearly}.
We illustrate such cases using the $\RSD$ problem namely the syndrome decoding problem in the rank metric setting.
As previously, let us consider $M$ instances such that $\sk = (\bm{x}_i)_{i \in \intoneto{M}}$ and $\pk = (\bm{H}, (\bm{y}^{\top}_i = \bm{H} \bm{x}^{\top}_i)_{i \in \intoneto{M}})$ where $\rw{\bm{x}_i} = \omega$.
Moreover, if we choose the $\bm{x}_i$ such that they all share the same support $E$, then we end up working with the $\RSL$ problem.
The $\RSL$ problem can be seen as a generalization of the $\RSD$ problem whose security has been studied in \cite{GHPT17, DT18, BB21}.
Due to the properties of the rank metric, any linear combination of $\bm{x}_i$ has a support $E' \subseteq E$ (with $E' = E$ with good probability) and a rank weight equal to $\omega' \leq \omega$ (with $\omega' = \omega$ with good probability).
One can leverage this property to design more efficient protocols by substituting a random choice of a $\bm{x}_i$ value by a random choice of a linear combination of the $\bm{x}_i$ values.
Doing so, a random choice amongst $M$ values is substitued by a random choice amongst $q^M$ values thus amplifying the considered challenge space size exponentially rather than linearly while still requiring the public key size to be increased only linearly by the factor $M$.
This makes our amplification technique quite efficient as a mitigation strategy against the \cite{KZ20} attack as suggested by the numbers presented in Section~\ref{sec:param-rsd}.
Indeed, one can see that amplifying the first challenge space size exponentially reduces the number of repetitions $\tau$ from $37$ to $23$ namely an improvement approximately greater than $35\%$.

%\vspace{\baselineskip}
%\noindent \emph{A remark on hard problems and their variants.} 
%While variants of the $\SD$ or $\RSD$ problems such as $\QCSD$ and $\IRSL$ have been studied by the code-based cryptography community, similar variants for other problems may not have received as much attention.
%We believe that this work may motivate additional studies in this direction.
%Some of these variants such as $\HMQ$ with cyclicity as well as $\HMQ$ and $\IPKP$ when given several instances are discussed later and will be explored in subsequent work.

%--------------------------------------------------------------------%
\section{New PoK with Helper for the $\SD$ problem} \label{sec:pok-sd1}
%--------------------------------------------------------------------%

The first efficient PoK for the $\SD$ problem was introduced by Stern in \cite{Stern93}.
Over the years, many variants and improvements have been proposed in \cite{Veron97, AGS11, CVE11} for instance.
Several PoK achieving an arbitrarily small soundness error have been proposed recently in \cite{GPS21, FJR21, BGKM22, FJR22}.
Most of these schemes (with the notable exception of \cite{FJR22}) relies on permutations to mask the secret value $\bm{x}$ associated to the considered $\SD$ problem.
Hereafter, we present a new PoK with Helper that outperforms previous permutation-based schemes.
It is described in Figure~\ref{fig:pok4-h} and encompasses several ideas from \cite{BGKM22}, \cite{FJR21} and \cite{BGKS21}.
Indeed, our protocol follows the same paradigm as PoK~2 from \cite{BGKM22} as it introduces several permutations $(\pi_i)_{i \in \intoneto{N}}$ and masks the secret $\bm{x}$ using the permutation $\pi_{\alpha}$ for $\alpha \in \intoneto{N}$ while revealing the other permutations $(\pi_i)_{i \in \intoneto{N} \setminus \alpha}$.
A notable difference comes from the fact that it also leverages the shared permutations paradigm from \cite{FJR21} where the permutations $(\pi_i)_{i \in \intoneto{N}}$ are nested around a global permutation $\pi$ such that $\pi = \pi_{N} \circ \cdots \circ \pi_1$.
As a consequence, our protocol masks the secret $\bm{x}$ using $\pi_{\alpha} \circ \cdots \circ \pi_1[\bm{x}]$ for $\alpha \in \intoneto{N}$.
This differs from \cite{FJR21} where the secret is masked by $\pi[\bm{x}]$ and where $\pi[\bm{u} + \bm{x}] + \bm{v}$ is computed without revealing $\pi$.
This difference allows us to perform a \emph{cut-and-choose with a meet in the middle approach} where the recurrence relation related to $\pi$ is used in both directions rather than just one.
Thanks to this new property, our protocol can benefit from an optimization introduced in \cite{BGKS21} that allows to substitute a vector by a random seed.
This improvement is of great importance as it implies that our PoK size scales with $1.5 \, \tau \, n$ (where $\tau$ is the number of repetitions required to achieve a negligible soundness error and $n$ is the vector length considered within the $\SD$ instance) while all previous protocols feature sizes scaling with a factor $2 \, \tau \, n$ instead.
As a consequence, our protocol outperforms the protocols from \cite{BGKM22} and \cite{FJR21}.
We defer the reader to Tables~\ref{table:param1} and~\ref{table:sd-ft} for a comparison with existing protocols including the recent \cite{FJR22} proposal.

\begin{center}
  \resizebox{1\textwidth}{!}{\pseudocode{%
    \hspace{440pt} \\[-\baselineskip][\hline]\\[-8pt]
    \underline{\pcalgostyle{Inputs~\&~Public~Data}} \\
    w = \bm{x}, ~ x = (\bm{H}, \bm{y}) \\[\baselineskip]
    \underline{\helperf(x)} \\[0.1\baselineskip]
    \theta \sampler \bit^{\lambda}, ~ \xi \sampler \bit^{\lambda} \\
    \pcfor i \in \intoneto{N} \pcdo \\
    \pcind \theta_{i} \samples{\theta} \bit^{\lambda}, ~ \phi_{i} \samples{\theta_i} \bit^{\lambda}, 
    ~ \pi_i \samples{\phi_i} \hperm{n}, ~ \bm{v}_i \samples{\phi_i} \Ftn, 
    ~ r_{1, i} \samples{\theta_i} \bit^{\lambda}, ~ \com_{1, i} = \commith{r_{1,i}, \, \phi_{i}} \\
    \pcend \\
    \pi = \pi_N \circ \cdots \circ \pi_1,
    ~ \bm{v} = \bm{v}_N + \sum\nolimits_{i \in \intoneto{N - 1}} \pi_N \circ \cdots \circ \pi_{i + 1}[\bm{v}_i],
    ~ \bm{r} \samples{\xi} \Ftn, ~ \bm{u} = \pi^{-1}[\bm{r} - \bm{v}] \\
    \com_1 = \commitb{\bm{H} \bm{u} \, || \, \pi[\bm{u}]+ \bm{v} \, || \, (\com_{1,i})_{i \in \intoneto{N}}} \\[0.5\baselineskip]
    \tx{Send } (\theta, \xi) \tx{ to the } \proverf \tx{ and } \com_1 \tx{ to the } \verifierf
    ~\\[3pt][\hline]\\
    \underline{\prover_1(w, x, \theta, \xi)} \\
    \tx{Compute } (\theta_i, \pi_i, \bm{v}_i)_{i \in \intoneto{N}} \tx{ and } \bm{u} \tx{ from } (\theta, \xi) \\
    \bm{s}_0 = \bm{u} + \bm{x} \\
    \pcfor i \in [1, N] \pcdo \\
    \pcind \bm{s}_i = \pi_i[\bm{s}_{i - 1}] + \bm{v}_i \\ 
    \pcend \\
    \com_2 = \commitb{\bm{u} + \bm{x} \, || \, (\bm{s}_i)_{i \in \intoneto{N}}} \\[\baselineskip]
    \underline{\verifier_1(x, \com_1)} \\
    \alpha \sampler \intoneto{N} \\[\baselineskip]
    \underline{\prover_2(w, x, \theta, \xi, \alpha)} \\
    \bm{z}_1 = \bm{u} + \bm{x}, ~ z_2 = (\theta_{i})_{i \in \intoneto{N} \setminus \alpha},
    ~ z_3 = \xi, ~ \bm{z}_4 = \pi_{\alpha} \circ \dots \circ \pi_1[\bm{x}] \\
    \rsp = (\bm{z}_1, z_2, z_3, \bm{z}_4, \com_{1, \alpha}) \\[\baselineskip]
    \underline{\verifier_2(x, \com_1, \com_2, \alpha, \rsp)} \\
    \tx{Compute } (\bar{\phi_i}, \bar{r}_{1,i}, \bar{\pi}_i, \bar{\bm{v}}_i)_{i \in \intoneto{N} \setminus \alpha} \tx{ from } z_2
    \tx{ and } \bar{\bm{t}}_N = \bar{\bm{r}} \tx{ from } z_3 \\
    \pcfor i \in \{N, \ldots, \alpha + 1\} \pcdo \\
    \pcind \bar{\bm{t}}_{i-1} = \bar{\pi}^{-1}_i[\bar{\bm{t}}_i - \bar{\bm{v}}_i] \\
    \pcend \\
    \bar{\bm{s}}_0 = \bm{z}_1, ~ \bar{\bm{s}}_{\alpha} = \bar{\bm{t}}_\alpha + \bm{z}_4, ~ \bar{\com}_{1, \alpha} = \com_{1, \alpha} \\
    \pcfor i \in [1, N] \setminus \alpha \pcdo \\
    \pcind \bar{\bm{s}}_i = \bar{\pi}_i[\bar{\bm{s}}_{i - 1}] + \bar{\bm{v}}_i \\ 
    \pcind \bar{\com}_{1,i} = \commith{\bar{r}_{1,i}, \, \bar{\phi_i}} \\
    \pcend \\
    b_1 \samplen \big( \com_{1} = \commitb{\bm{H} \bm{z}_1 - \bm{y} \, || \, \bar{\bm{r}} \, || \, (\bar{\com}_{1,i})_{i \in \intoneto{N}}} \big) \\
    b_2 \samplen \big( \com_2 = \commitb{\bm{z}_1 \, || \, (\bar{\bm{s}}_i)_{i \in \intoneto{N}}} \big) \\ 
    b_3 \samplen \big( \hw{\bm{z}_4} = \omega \big) \\
    \pcreturn b_1 \wedge b_2 \wedge b_3
    ~\\[3pt][\hline]\\[-2\baselineskip]
  }}

  \captionof{figure}{PoK with Helper for the $\SD$ problem over $\Ft$ \label{fig:pok4-h}}
\end{center}

\newpage

\begin{theorem} \label{thm:pok-sd1}
If the hash function used is collision-resistant and if the commitment scheme used is binding and hiding, then the protocol depicted in Figure~\ref{fig:pok4-h} is an honest-verifier zero-knowledge PoK with Helper for the $\SD$ problem over $\Ft$ with soundness error $1/N$.   
\end{theorem}

\begin{proof}
Proof of Theorem~\ref{thm:pok-sd1} can be found in Appendix~\ref{app:pok-sd1}.
\end{proof}

%--------------------------------------------------------------------%
\section{New PoK related to the $\PKP$, $\SD$ and $\RSD$ problems} \label{sec:pokls}
%--------------------------------------------------------------------%

\subsection{PoK leveraging structure for the $\IPKP$ problem} \label{sec:pok-pkp2}

We present in Figure~\ref{fig:pkp} a PoK leveraging structure for the $\IPKP$ problem using linearity over $\Fq$ along with the shared permutation from \cite{FJR21}.

\begin{center}
  \resizebox{1\textwidth}{!}{
    \pseudocode{%
      \hspace{440pt} \\[-2\baselineskip][\hline]\\[-8pt]
      \underline{\pcalgostyle{Inputs~\&~Public~Data}} \\
      w = \pi, ~ x = (\bm{H}, \bm{x}, \bm{y}) \\
      \chsps = \Fq^*, ~ \chs = \kappa \\[\baselineskip]
      \underline{\prover_1(w, x)} \\
      \theta \sampler \bit^{\lambda} \\
      \pcfor i \in \{N,\ldots, 1 \} \pcdo \\
      \pcind \theta_{i} \samples{\theta} \bit^{\lambda}, ~ \phi_{i} \samples{\theta_i} \bit^{\lambda} \\
      \pcind \pcif i \neq 1 \pcdo \\
      \pcind \pcind \pi_i \samples{\phi_i} \hperm{n}, ~ \bm{v}_i \samples{\phi_i} \Fq^n,
      ~ r_{1, i} \samples{\theta_i} \bit^{\lambda}, ~ \com_{1, i} = \commith{r_{1,i}, \, \phi_i} \\
      \pcind \pcelse \\
      \pcind \pcind \pi_1 = \pi^{-1}_2 \circ \cdots \circ \pi^{-1}_N \circ \pi, ~ \bm{v}_1 \samples{\phi_1} \Fq^n,
      ~ r_{1, 1} \samples{\theta_1} \bit^{\lambda}, ~ \com_{1, 1} = \commith{r_{1,1}, \, \pi_1 \, || \, \phi_1} \\
      \pcind \pcend \\
      \pcend \\
      \bm{v} = \bm{v}_N + \sum\nolimits_{i \in \intoneto{N - 1}} \pi_N \circ \cdots \circ \pi_{i + 1}[\bm{v}_i] \\
      \com_1 = \commitb{\bm{H} \bm{v} \, || \, (\com_{1,i})_{i \in \intoneto{N}}} \\[\baselineskip] 
      \underline{\prover_2(w, x, \kappa)} \\
      \bm{s}_0 = \kappa \cdot \bm{x} \\
      \pcfor i \in [1, N] \pcdo \\
      \pcind \bm{s}_i = \pi_i[\bm{s}_{i - 1}] + \bm{v}_i \\ 
      \pcend \\
      \com_2 = \commitb{(\bm{s}_i)_{i \in \intoneto{N}}} \\[\baselineskip]
      \underline{\prover_2(w, x, \kappa, \alpha)} \\
      \bm{z}_1 = \bm{s}_{\alpha} \\
      \pcif \alpha \neq 1 \pcdo \\
      \pcind z_2 = \pi_1 \, || \, (\theta_{i})_{i \in \intoneto{N} \setminus \alpha} \\
      \pcelse \\
      \pcind z_2 = (\theta_{i})_{i \in \intoneto{N} \setminus \alpha} \\
      \pcend \\
      \rsp = (\bm{z}_1, z_2, \com_{1,\alpha})
      ~\\[3pt][\hline]\\[-2\baselineskip]
  }}
\end{center}

\begin{center}
  \resizebox{1\textwidth}{!}{
    \pseudocode{%
      \hspace{440pt} \\[-2\baselineskip][\hline]\\[-8pt]
      \underline{\verifier(x, \com_1, \kappa, \com_2, \alpha, \rsp)} \\
      \tx{Compute } (\bar{\phi_i}, \bar{r}_{1,i}, \bar{\pi}_i, \bar{\bm{v}}_i)_{i \in \intoneto{N} \setminus \alpha} \tx{ from } z_2 \\
      \bar{\bm{s}}_0 = \kappa \cdot \bm{x}, ~ \bar{\bm{s}}_{\alpha} = \bm{z}_1, ~ \bar{\com}_{1, \alpha} = \com_{1, \alpha} \\
      \pcfor i \in [1, N] \setminus \alpha \pcdo \\
      \pcind \bar{\bm{s}}_i = \bar{\pi}_i[\bar{\bm{s}}_{i - 1}] + \bar{\bm{v}}_i \\ 
      \pcind \pcif i \neq 1 \pcdo \\
      \pcind \pcind \bar{\com}_{1,i} = \commith{\bar{r}_{1,i}, \, \bar{\phi}_i} \\
      \pcind \pcelse \\
      \pcind \pcind \bar{\com}_{1,1} = \commith{\bar{r}_{1,1}, \, \bar{\pi}_1 \, || \, \bar{\phi}_1} \\
      \pcind \pcend \\
      \pcend \\
      b_1 \samplen \big( \com_{1} = \commitb{\bm{H} \bar{\bm{s}}_N - \kappa \cdot \bm{y} \, || \, (\com_{1,i})_{i \in \intoneto{N}}} \big) \\
      b_2 \samplen \big( \com_2 = \commitb{(\bar{\bm{s}}_i)_{i \in \intoneto{N}}} \big) \\ 
      \pcreturn b_1 \wedge b_2 
      ~\\[3pt][\hline]\\[-2\baselineskip]
    }
  }
  \captionof{figure}{PoK leveraging structure for the $\IPKP$ problem \label{fig:pkp}}
\end{center}

\begin{theorem} \label{thm:pok-pkp2}
  If the hash function used is collision-resistant and if the commitment scheme used is binding and hiding, then the protocol depicted in Figure~\ref{fig:pkp} is an honest-verifier zero-knowledge PoK for the $\IPKP$ problem with soundness error equal to $\frac{1}{N} + \frac{N - 1}{N \cdot (q - 1)}$.
\end{theorem}

\begin{proof}
Proof of Theorem~\ref{thm:pok-pkp2} can be found in Appendix~\ref{app:pok-pkp2}.
\end{proof}

\subsection{PoK leveraging structure for the $\QCSD$ problem over $\Ft$} \label{sec:pok-sd2}

We apply results from Sections~\ref{sec:paradigm} and~\ref{sec:amplification} on top of the PoK with Helper from Section~\ref{sec:pok-sd1}.
This PoK leverages quasi-cyclicity over $\Ft$ and is depicted in Figure~\ref{fig:pok4-qc}.
%In addition, a variant using linearity over $\Fq$ is described in Appendix~\ref{app:pok-sd2-fq}.

%We provide two variants that respectively leverage quasi-cyclicity over $\Ft$ and linearity over $\Fq$.
%This results in two PoK based on the $\QCSD$ over $\Ft$ as well as the $\SD$ over $\Fq$ that are respectively described in Figures~\ref{fig:pok4-qc} and~\ref{fig:pok4-fq}.
%Our first variant shares some similarity with the Quasi-Cyclic Stern protocol from \cite{BGKS21} that improves the initial Stern protocol using quasi-cyclicity.
%Indeed, it is based on the $\DiffSD$ problem and its security relies on the reduction from the $\QCSD$ problem to the $\DiffSD$ problem.
%In practice, one has to increase the considered parameters in order for this reduction to hold.

\begin{center}
  \resizebox{1\textwidth}{!}{
    \pseudocode{%
      \hspace{440pt} \\[-2\baselineskip][\hline]\\[-8pt]
      \underline{\pcalgostyle{Inputs~\&~Public~Data}} \\
      w = (\bm{x}_i)_{i \in \intoneto M}, ~ x = (\bm{H}, (\bm{y}_i)_{i \in \intoneto{M}}), \\
      \chsps = \intoneto{M} \times \intoneto{k}, ~ \chs = (\mu, \kappa) \\[\baselineskip]
      \underline{\prover_1(w, x)} \\
      \theta \sampler \bit^{\lambda}, ~ \xi \sampler \bit^{\lambda} \\
      \pcfor i \in \intoneto{N} \pcdo \\
      \pcind \theta_{i} \samples{\theta} \bit^{\lambda}, ~ \phi_{i} \samples{\theta_i} \bit^{\lambda}, 
      ~ \pi_i \samples{\phi_i} \hperm{n}, ~ \bm{v}_i \samples{\phi_i} \Ftn,
      ~ r_{1, i} \samples{\theta_i} \bit^{\lambda}, ~ \com_{1, i} = \commith{r_{1,i}, \, \phi_{i}} \\
      \pcend \\
      \pi = \pi_N \circ \cdots \circ \pi_1, ~ \bm{v} = \bm{v}_N + \sum\nolimits_{i \in \intoneto{N - 1}} \pi_N \circ \cdots \circ \pi_{i + 1}[\bm{v}_i],
      ~ \bm{r} \samples{\xi} \Ftn, ~ \bm{u} = \pi^{-1}[\bm{r} - \bm{v}] \\
      \com_1 = \commitb{\bm{H} \bm{u} \, || \, \pi[\bm{u}]+ \bm{v} \, || \, (\com_{1,i})_{i \in \intoneto{N}}} \\[\baselineskip] 
      \underline{\prover_2(w, x, \mu, \kappa)} \\
      \bm{x}_{\mu, \kappa} = \rot_{\kappa}(\bm{x}_{\mu}), ~ \bm{s}_0 = \bm{u} + \bm{x}_{\mu, \kappa} \\
      \pcfor i \in [1, N] \pcdo \\
      \pcind \bm{s}_i = \pi_i[\bm{s}_{i - 1}] + \bm{v}_i \\ 
      \pcend \\
      \com_2 = \commitb{\bm{u} + \bm{x}_{\mu, \kappa} \, || \, (\bm{s}_i)_{i \in \intoneto{N}}}
      ~\\[3pt][\hline]\\[-2\baselineskip]
  }}
\end{center}

\begin{center}
  \resizebox{1\textwidth}{!}{
    \pseudocode{%
      \hspace{440pt} \\[-2\baselineskip][\hline]\\[-8pt]
      \underline{\prover_3(w, x, \mu, \kappa, \alpha)} \\
      \bm{z}_1 = \bm{u} + \bm{x}_{\mu, \kappa}, ~ z_2 = (\theta_{i})_{i \in \intoneto{N} \setminus \alpha},  
      ~ z_3 = \xi, ~ \bm{z}_4 = \pi_{\alpha} \circ \dots \circ \pi_1[\bm{x}_{\mu, \kappa}] \\
      \rsp = (\bm{z}_1, z_2, z_3, \bm{z}_4, \com_{1,\alpha}) \\[\baselineskip]
      \underline{\verifier(x, \com_1, (\mu, \kappa), \com_2, \alpha, \rsp)} \\
      \tx{Compute } (\bar{\phi_i}, \bar{r}_{1,i}, \bar{\pi}_i, \bar{\bm{v}}_i)_{i \in \intoneto{N} \setminus \alpha} \tx{ from } z_2 
      \tx{ and } \bar{\bm{t}}_N = \bar{\bm{r}} \tx{ from } z_3 \\
      \pcfor i \in \{N, \ldots, \alpha + 1\} \pcdo \\
      \pcind \bar{\bm{t}}_{i-1} = \bar{\pi}^{-1}_i[\bar{\bm{t}}_i - \bar{\bm{v}}_i] \\
      \pcend \\
      \bar{\bm{s}}_0 = \bm{z}_1, ~ \bar{\bm{s}}_{\alpha} = \bar{\bm{t}}_\alpha + \bm{z}_4, ~ \bar{\com}_{1, \alpha} = \com_{1, \alpha} \\
      \pcfor i \in [1, N] \setminus \alpha \pcdo \\
      \pcind \bar{\bm{s}}_i = \bar{\pi}_i[\bar{\bm{s}}_{i - 1}] + \bar{\bm{v}}_i, 
      ~ \bar{\com}_{1,i} = \commith{\bar{r}_{1,i}, \, \bar{\phi_i}} \\
      \pcend \\
      b_1 \samplen \big( \com_{1} = \commitb{\bm{H} \bm{z}_1 - \rot_{\kappa}(\bm{y}_{\mu}) \, || \, \bar{\bm{r}} \, || \, (\bar{\com}_{1,i})_{i \in \intoneto{N}}} \big) \\
      b_2 \samplen \big( \com_2 = \commitb{\bm{z}_1 \, || \, (\bar{\bm{s}}_i)_{i \in \intoneto{N}}} \big) \\
      b_3 \samplen \big( \hw{\bm{z}_4} = \omega \big) \\
      \pcreturn b_1 \wedge b_2 \wedge b_3
      ~\\[3pt][\hline]\\[-2\baselineskip]
    }
  }
  \captionof{figure}{PoK leveraging structure for the $\QCSD(n, k, w, M)$ problem over $\Ft$ \label{fig:pok4-qc}}
\end{center}

\begin{theorem} \label{thm:pok-sd2-qc}
If the hash function used is collision-resistant and if the commitment scheme used is binding and hiding, then the protocol depicted in Figure~\ref{fig:pok4-qc} is an honest-verifier zero-knowledge PoK for the $\QCSD(n, k, w, M)$ problem with soundness error equal to $\frac{1}{N} + \frac{(N - 1)(\Delta - 1)}{N \cdot M \cdot k}$ for some parameter $\Delta$.
\end{theorem}

\begin{proof}
Proof of Theorem~\ref{thm:pok-sd2-qc} can be found in Appendix~\ref{app:pok-sd2-qc}.
\end{proof}

\noindent \emph{Remark.} The parameter $\Delta$ is related to the security of the $\DiffSD$ problem.
The $\QCSD$ problem reduces to the intermediary $\DiffSD$ problem for sufficient large values of $\Delta$ as showed in Appendix~\ref{app:pok-sd2-qc}.
When $\Delta$ is not large enough, the security reduction does not hold which enable additional cheating strategies for the adversary hence impact the soundness of the protocol.

%\noindent \textbf{Variants with other metrics.}
%\textcolor{blue}{
%Our protocol can be adapted to other metrics such as the Lee metric and the Hamming metric along with secrets of large weight by using isometries tailored to the considered metric as explained in \cite{BGKS21}.
%}

%\noindent \textbf{Additional optimization.}
%One should note that the protocol described in Sections~\ref{sec:pok-sd1},~\ref{sec:pok-sd2} and~\ref{sec:pok-rsd2} induce an overhead on the signing part.
%In order to mitigate it, one could consider a variant that use the cut and choose with meet in the middle approach from bottom to top or top to bottom adaptively depending if $\alpha$ is closer to $1$ or closer to $N$.

%\vspace{\baselineskip}
%\noindent \textbf{Additional variant.} 
%Our new paradigm apply straighforwardly to \cite{FJR21} too which lead to another trade-off between signature size and performances.

\subsection{PoK leveraging structure for the $\IRSL$ problem} \label{sec:pok-rsd2}

One can adapt the protocols described in Sections~\ref{sec:pok-sd1} and~\ref{sec:pok-sd2} to the rank metric setting by replacing permutations by isometries for the rank metric.
Doing so, one get a PoK with Helper for the $\RSD$ problem as well as a PoK leveraging structure for the $\IRSD$ problem.
In addition, one can apply the challenge space amplification technique presented in Section~\ref{sec:amplification} in order to get a PoK leveraging structure for $\IRSL$ problem as depicted in Figure~\ref{fig:pok-rsd2}.
In practice, one has to take into account the cases where $\omega' < \omega$ mentionned in Section~\ref{sec:amplification} however we omit such cases in Figure~\ref{fig:pok-rsd2} for conciseness.

\begin{center}
  \resizebox{1\textwidth}{!}{
    \pseudocode{%
      \hspace{440pt} \\[-2\baselineskip][\hline]\\[-8pt]
      \underline{\pcalgostyle{Inputs~\&~Public~Data}} \\
      w = (\bm{x}_i)_{i \in \intoneto M}, ~ x = (\bm{H}, (\bm{y}_i)_{i \in \intoneto{M}}), \\
      \chsps = (\Fq)^{Mk} \setminus (0, \cdots, 0), ~ \chs = \gamma = (\gamma_{i, j})_{i \in \intoneto{M}, j \in \intoneto{k}}
      ~\\[3pt][\hline]\\[-2\baselineskip]
  }}
\end{center}

\begin{center}
  \resizebox{1\textwidth}{!}{
    \pseudocode{%
      \hspace{440pt} \\[-2\baselineskip][\hline]\\[-8pt]
      \underline{\prover_1(w, x)} \\
      \theta \sampler \bit^{\lambda}, ~ \xi \sampler \bit^{\lambda} \\
      \pcfor i \in \intoneto{N} \pcdo \\
      \pcind \theta_{i} \samples{\theta} \bit^{\lambda}, ~ \phi_{i} \samples{\theta_i} \bit^{\lambda},
      ~ \bm{P}_i \samples{\phi_i} \GLFq{m}, ~ \bm{Q}_i \samples{\phi_i} \GLFq{n}, ~ \bm{v}_i \samples{\phi_i} \Fqm^n,
      ~ r_{1, i} \samples{\theta_i} \bit^{\lambda}, ~ \com_{1, i} = \commith{r_{1,i}, \, \phi_{i}} \\
      \pcend \\
      \bm{P} = \bm{P}_N \times \cdots \times \bm{P}_1, ~ \bm{Q} = \bm{Q}_N \times \cdots \times \bm{Q}_1,
      ~ \bm{v} = \bm{v}_N + \sum\nolimits_{i \in \intoneto{N - 1}} \bm{P}_N \times \cdots \times \bm{P}_{i + 1} \cdot \bm{v}_i \cdot \bm{Q}_{i + 1} \times \cdots \times \bm{Q}_N \\
      \bm{r} \samples{\xi} \Ftn, ~ \bm{u} = \bm{P}^{-1} \cdot (\bm{r} - \bm{v}) \cdot \bm{Q}^{-1} \\
      \com_1 = \commitb{\bm{H} \bm{u} \, || \, \bm{P} \cdot \bm{u} \cdot \bm{Q} + \bm{v} \, || \, (\com_{1,i})_{i \in \intoneto{N}}} \\[\baselineskip]
      \underline{\prover_2(w, x, \gamma)} \\
      \bm{x}_{\gamma} = \sum\nolimits_{(i, j) \in \intoneto{M} \times \intoneto{k}} \gamma_{i, j} \cdot \rot_{j}(\bm{x}_{i}),
      ~ \bm{s}_0 = \bm{u} + \bm{x}_{\gamma} \\
      \pcfor i \in [1, N] \pcdo \\
      \pcind \bm{s}_i = \bm{P}_i \cdot \bm{s}_{i - 1} \cdot \bm{Q}_i + \bm{v}_i \\ 
      \pcend \\
      \com_2 = \commitb{\bm{u} + \bm{x}_{\gamma} \, || \, (\bm{s}_i)_{i \in \intoneto{N}}} \\[\baselineskip]
      \underline{\prover_3(w, x, \gamma, \alpha)} \\
      \bm{z}_1 = \bm{u} + \bm{x}_{\gamma}, ~ z_2 = (\theta_{i})_{i \in \intoneto{N} \setminus \alpha},
      ~ z_3 = \xi, ~ \bm{z}_4 = \bm{P}_{\alpha} \times \dots \times \bm{P}_1 \cdot \bm{x}_{\gamma} \cdot \bm{Q}_{1} \times \cdots \times \bm{Q}_{\alpha} \\
      \rsp = (\bm{z}_1, z_2, z_3, \bm{z}_4, \com_{1,\alpha}) \\[\baselineskip]
      \underline{\verifier(x, \com_1, \gamma, \com_2, \alpha, \rsp)} \\
      \tx{Compute } (\bar{\phi_i}, \bar{r}_{1,i}, \bar{\bm{P}}_i, \bar{\bm{Q}}_i, \bar{\bm{v}}_i)_{i \in \intoneto{N} \setminus \alpha} \tx{ from } z_2 
      \tx{ and } \bar{\bm{t}}_N = \bar{\bm{r}} \tx{ from } z_3 \\
      \pcfor i \in \{N, \ldots, \alpha + 1\} \pcdo \\
      \pcind \bar{\bm{t}}_{i-1} = \bar{\bm{P}}^{-1}_i \cdot (\bar{\bm{t}}_i - \bar{\bm{v}}_i) \cdot \bar{\bm{Q}}^{-1}_i \\
      \pcend \\
      \bar{\bm{s}}_0 = \bm{z}_1, ~ \bar{\bm{s}}_{\alpha} = \bar{\bm{t}}_\alpha + \bm{z}_4, ~ \bar{\com}_{1, \alpha} = \com_{1, \alpha} \\
      \pcfor i \in [1, N] \setminus \alpha \pcdo \\
      \pcind \bar{\bm{s}}_i = \bar{\bm{P}}_i \cdot \bar{\bm{s}}_{i - 1} \cdot \bar{\bm{Q}}_i + \bar{\bm{v}}_i \\ 
      \pcind \bar{\com}_{1,i} = \commith{\bar{r}_{1,i}, \, \bar{\phi_i}} \\
      \pcend \\
      b_1 \samplen \big( \com_{1} = \commitb{\bm{H} \bm{z}_1 - \sum\nolimits_{(i, j) \in \intoneto{M} \times \intoneto{k}} \gamma_{i, j} \cdot \rot_{j}(\bm{y}_{i}) \, || \, \bar{\bm{r}} \, || \, (\bar{\com}_{1,i})_{i \in \intoneto{N}}} \big) \\
      b_2 \samplen \big( \com_2 = \commitb{\bm{z}_1 \, || \, (\bar{\bm{s}}_i)_{i \in \intoneto{N}}} \big) \\ 
      b_3 \samplen \big( \hw{\bm{z}_4} = \omega \big) \\
      \pcreturn b_1 \wedge b_2 \wedge b_3
      ~\\[3pt][\hline]\\[-2\baselineskip]
    }
  }
  \captionof{figure}{PoK leveraging structure for the $\IRSL$ problem \label{fig:pok-rsd2}}
\end{center}

\begin{theorem} \label{thm:pok-rsd2}
If the hash function used is collision-resistant and if the commitment scheme used is binding and hiding, then the protocol depicted in Figure~\ref{fig:pok-rsd2} is an honest-verifier zero-knowledge PoK for the $\IRSL$ problem with soundness error equal to $\frac{1}{N} + \frac{(N - 1)(\Delta - 1)}{N (q^{Mk} - 1)}$ for some parameter $\Delta$.
\end{theorem}

\begin{proof}
Proof of Theorem~\ref{thm:pok-rsd2} can be found in Appendix~\ref{app:pok-rsd2}.
\end{proof}

%\subsection{PoK leveraging structure for the $\HMQ$ problem} \label{sec:pok-mq2}

%\input{pok-hmq}

\FloatBarrier

%--------------------------------------------------------------------%
\section{Resulting signatures and comparison} \label{sec:comparison}
%--------------------------------------------------------------------%

PoK can be transformed into signature schemes using the Fiat-Shamir transform \cite{FS}.
Several optimizations can be employed in the process, we defer the interested reader to previous work such as \cite{KKW18, Beullens20, GPS21, BGKM22} for additional details.
Hereafter, we keep the inherent 5-round structure of our PoK (except for the one from Section~\ref{sec:pok-sd1} that is collapsed to 3-round) hence parameters are chosen taking into account the attack from \cite{KZ20}.
Moreover, we only consider parameters for $\lambda = 128$ bits of security.
%While some of the problems used to construct our PoK have been well studied, other have not received much attention in the literature.
%As a consequence, parameters for these problems can not be chosen with the same level of confidence and are mainly provided to illustrate the potential of our new paradigm.
%We highlight these cases to avoid any ambiguity and believe that this work may motivate further research in these areas.
The commitments are instantiated using hash functions along with some randomness.
For the signatures, random salts are added to the hash functions.

\subsection{Signatures based on PoK related to the $\PKP$ problem}

The signature size of our protocol from Section~\ref{sec:pok-pkp2} is detailed in Table~\ref{table:pkp1}.
Table~\ref{table:pkp2} provides a comparison with respect to other PoK for the $\PKP$ problem.
The complexity of the $\PKP$ problem has been studied in \cite{PKP-complexity}.
We consider parameters from \cite{Beullens20} for our comparison namely $(q = 997$, $n = 61$, $m = 28)$.
In addition, we use both $(N = 32, \tau = 42)$ and $(N = 256, \tau = 31)$.

%\vspace{\baselineskip}
%\textcolor{blue}{
%\noindent \textbf{Future investigations.}
%By analogy to the decoding one out of many setting for the $\SD$ problem as well as the $\HMQP$ problem, it would be interesting to study the complexity of the $\IPKP^+$ problem that would be defined as a variant of the $\IPKP$ problem where several pairs of vectors $(\bm{x}_i, \bm{y}_i)$ are given.
%This would lead to an improvement over our PoK leveraging structure for the $\IPKP$ problem that will be studied in subsequent work.
%}

\vspace{-0.5\baselineskip}

\begin{table}[H]
\begin{center}
{\setlength{\tabcolsep}{0.4em}
{\renewcommand{\arraystretch}{1.6}
{\scriptsize
  \begin{tabular}{|l|l|}
    \cline{2-2}
    \multicolumn{1}{c|}{} & \multicolumn{1}{c|}{Signature size} \\ 
    \hline 
    Our Work (Section \ref{sec:pok-pkp2})   & $5 \lambda + \tau \cdot (n \cdot \log_2(q) + n \cdot \log_2(n) + \lambda \cdot \log_2(N) + 2\lambda)$ \\ \hline 
  \end{tabular}
  \vspace{0.5\baselineskip}
  \caption{Signature sizes for our $\PKP$ based construction} \label{table:pkp1}
}}}
\end{center}
\end{table}

\vspace{-4.5\baselineskip}

\begin{table}[H]
\begin{center}
{\setlength{\tabcolsep}{0.4em}
{\renewcommand{\arraystretch}{1.6}
{\scriptsize
  \begin{tabular}{|l|c|c|c|c|l|}
    \cline{2-6}
    \multicolumn{1}{c|}{} & Type & $\pk$ & $\sigma$ & Structure & Security Assumption \\ \hline

    \multirow{2}{*}{SUSHYFISH \cite{Beullens20}}           & Fast  & 0.1 kB & 18.1 kB & \multirow{2}{*}{3-round} & \multirow{2}{*}{$\IPKP$} \\ \cline{2-4}
                                                           & Short & 0.1 kB & 12.1 kB &                          & \\ \hline

    \multirow{2}{*}{Our Work (Section \ref{sec:pok-pkp2})} & Fast  & 0.1 kB &  10.0 kB & \multirow{2}{*}{5-round} & \multirow{2}{*}{$\IPKP$} \\ \cline{2-4}
                                                           & Short & 0.1 kB &  8.9 kB &                          & \\ \hline

  \end{tabular}
  \vspace{0.5\baselineskip}
  \caption{Signatures based on $\PKP$ for $\lambda = 128$ (sorted by decreasing size)} \label{table:pkp2}
}}}
\end{center}
\end{table}
\vspace{-3\baselineskip}

\subsection{Signatures based on PoK related to the $\SD$ problem} \label{sec:param-sd}

\noindent Table~\ref{table:param1} compares signature sizes of our new protocol with respect to existing ones.
One can see that our signatures scale with a factor $1.5n \cdot \tau$ which bring an improvement with respect to previous (comparable) schemes that scale with a factor $2n \cdot \tau$.
Table~\ref{table:sd-ft} provides a comparison to other code-based signatures constructed from PoK for the $\SD$ problem.
Parameters are chosen taking into account attacks from \cite{BJMM12} and \cite{Sen11}.
For our protocol from Section~\ref{sec:pok-sd1}, we have used $(n = 1190, k = 595, \omega = 132)$ as well as $(N = 8, \tau = 49, M' = 187)$ and $(N = 32, \tau = 28, M' = 389)$.
In these cases, $M'$ is the parameter related to the beating parallel repetition from \cite{KKW18} namely $M'$ instances are prepared during the preprocessing step amongst which $\tau$ are actually executed.
For our protocol from Section~\ref{sec:pok-sd2}, we have used $(n = 1306, k = 653, \omega = 132, \Delta = 17)$ as well as $(N = 32, \tau = 42, M = 22)$ and $(N = 256, \tau = 33, M = 12)$.
%In addition, we have considered a variant that leverages our challenge space amplification technique.
%For this variant, parameters used are $(M = 16, N = 32, \tau = 37)$ and $(M = 11, N = 256, \tau = 27)$.
Numbers for \cite{FJR22} are from the original paper while numbers for \cite{FJR21} have been recomputed using the aforementioned parameters in order to provide a fair comparison.

\begin{table}[h]
\begin{center}
{\setlength{\tabcolsep}{0.4em}
{\renewcommand{\arraystretch}{1.6}
{\scriptsize
  \begin{tabular}{|l|l|}
    \cline{2-2}
    \multicolumn{1}{c|}{} & \multicolumn{1}{c|}{Signature size} \\ 
    \hline 
      
    \cite{BGKM22}                         & $2 \lambda + \tau \cdot (2n + 3 \lambda \cdot \log_2(N) + 3 \lambda \cdot \log_2(M/\tau))$ \\ \hline
    \cite{FJR21}                          & $2 \lambda + \tau \cdot (2n + \lambda \cdot \log_2(N) + 2\lambda + 3 \lambda \cdot \log_2(M/\tau))$ \\ \hline
    Our Work (Section \ref{sec:pok-sd1})  & $3 \lambda + \tau \cdot (1.5n               + \lambda \cdot \log_2(N) + 2\lambda + 3 \lambda \cdot \log_2(M/\tau))$ \\ \hline
    Our Work (Section \ref{sec:pok-sd2})  & $5 \lambda + \tau \cdot (1.5n + \lambda \cdot \log_2(N) + 2\lambda)$ \\ \hline 
  \end{tabular}
  \vspace{0.5\baselineskip}
  \caption{Signature sizes (sorted by decreasing size)} \label{table:param1}
}}}
\end{center}
\end{table}

\vspace{-7\baselineskip}

\begin{table}[H]
\begin{center}
{\setlength{\tabcolsep}{0.4em}
{\renewcommand{\arraystretch}{1.6}
{\scriptsize
  \begin{tabular}{|l|c|c|c|c|l|}
    \cline{2-6}
    \multicolumn{1}{c|}{}  & Type & $\pk$ & $\sigma$ & Structure & Security Assumption \\ \hline

    \multirow{2}{*}{\cite{BGKM22}}                         & Fast  & 0.1 kB & 26.4 kB & \multirow{2}{*}{3-round} & \multirow{2}{*}{$\SD$ over $\Ft$} \\ \cline{2-4}
                                                           & Short & 0.1 kB & 20.5 kB &                          & \\ \hline

    \multirow{2}{*}{\cite{FJR21}}                          & Fast  & 0.1 kB & 23.3 kB & \multirow{2}{*}{3-round} & \multirow{2}{*}{$\SD$ over $\Ft$} \\ \cline{2-4}
                                                           & Short & 0.1 kB & 16.9 kB &                          & \\ \hline

    \multirow{2}{*}{Our Work (Section \ref{sec:pok-sd1})}  & Fast  & 0.1 kB & 19.6 kB & \multirow{2}{*}{3-round} & \multirow{2}{*}{$\SD$ over $\Ft$} \\ \cline{2-4}
                                                           & Short & 0.1 kB & 14.8 kB &                          & \\ \hline

    \multirow{2}{*}{Our Work (Section \ref{sec:pok-sd2})}  & Fast  & 1.8 kB & 15.1 kB & \multirow{2}{*}{5-round} & \multirow{2}{*}{$\QCSD$ over $\Ft$} \\ \cline{2-4}
                                                           & Short & 1.0 kB & 13.5 kB &                          & \\ \hline

    \multirow{4}{*}{\cite{FJR22}}                          & Fast  & 0.1 kB & 17.0 kB & \multirow{2}{*}{5-round} & \multirow{2}{*}{$\SD$ over $\Ft$} \\ \cline{2-4}
                                                           & Short & 0.1 kB & 11.8 kB &                          & \\ \cline{2-6}

                                                           & Fast  & 0.2 kB & 11.5 kB & \multirow{2}{*}{5-round} & \multirow{2}{*}{$\SD$ over $\Fq$} \\ \cline{2-4}
                                                           & Short & 0.2 kB &  8.3 kB &                          & \\ \hline
  \end{tabular}
  \vspace{0.5\baselineskip}
  \caption{Signatures based on $\SD$ for $\lambda = 128$ (sorted by decreasing size)} \label{table:sd-ft}
}}}
\end{center}
\end{table}

\vspace{-3\baselineskip}

\subsection{Signatures based on PoK related to the $\RSD$/$\RSL$ problem} \label{sec:param-rsd}

Parameters for our PoK based on the rank metric are chosen to resist best known attacks against $\RSD$ \cite{RSD-attack, BBC20} and $\RSL$ \cite{BB21, GHPT17, DT18}.
%Numbers for Durandal are from \cite{Durandal}.
For our protocol from Section~\ref{sec:pok-sd1}, we have used $(q = 2, m = 31, n = 30, k = 15, \omega = 9)$ as well as $(N = 8, \tau = 49, M' = 187)$ and $(N = 32, \tau = 28, M' = 389)$.
For our protocol from Section~\ref{sec:pok-rsd2}, we have used $(q = 2, m = 37, n = 34, k = 17, \omega = 9, \Delta = 10)$ as well as $(N = 32, \tau = 37)$ and $(N = 512, \tau = 25)$ for the variant relying on the $\IRSD$ problem.
In addition, we have used $(q = 2, m = 37, n = 34, k = 17, \omega = 10, \Delta = 40, M = 5)$ as well as $(N = 64, \tau = 23)$ and $(N = 1024, \tau = 14)$ for the variant relying on the $\IRSL$ problem.
%\textcolor{blue}{In addition, one could adapt our construction to create a PoK based on the (non-ideal) $\RSL$ problem whose sizes would be in between the ones based on the $\RSD$ and $\IRSL$ problems.}

\vspace{-\baselineskip}

\begin{table}[H]
\begin{center}
{\setlength{\tabcolsep}{0.4em}
{\renewcommand{\arraystretch}{1.6}
{\scriptsize
  \begin{tabular}{|l|l|}
    \cline{2-2}
    \multicolumn{1}{c|}{} & \multicolumn{1}{c|}{Signature size} \\ 
    \hline 
    Our Work (Section \ref{sec:pok-sd1})  & $3 \lambda + \tau \cdot (mn + \omega(m + n - \omega) + \lambda \cdot \log_2(N) + 2\lambda + 3 \lambda \cdot \log_2(M/\tau))$ \\ \hline
    Our Work (Section \ref{sec:pok-rsd2}) & $5 \lambda + \tau \cdot (mn + \omega(m + n - \omega) + \lambda \cdot \log_2(N) + 2\lambda)$ \\ \hline 
  \end{tabular}
  \vspace{0.5\baselineskip}
  \caption{Signature sizes for our $\RSD$ based constructions} \label{table:rsd1}
}}}
\end{center}
\end{table}

\vspace{-3\baselineskip}

\begin{table}[H]
\begin{center}
{\setlength{\tabcolsep}{0.4em}
{\renewcommand{\arraystretch}{1.6}
{\scriptsize
  \begin{tabular}{|l|c|c|c|c|l|}
    \cline{2-6}
    \multicolumn{1}{c|}{} & Type & $\pk$ & $\sigma$ & Structure & Security Assumption \\ \hline

    %Durandal                                              & -     & 15.3 kB & 4.1 kB  & -                        & $\IRSD$, $\RSL$, $\PSSI$ \\ \hline
    \cite{RankAGS19}                                       & -     & 0.2 kB & 22.5 kB  & 5-round                  & $\IRSD$ \\ \hline

    \multirow{2}{*}{Our Work (Section \ref{sec:pok-sd1})}  & Fast  & 0.1 kB  & 17.2 kB & \multirow{2}{*}{3-round} & \multirow{2}{*}{$\RSD$} \\ \cline{2-4}
                                                           & Short & 0.1 kB  & 13.5 kB &                          & \\ \hline

    \multirow{4}{*}{Our Work (Section \ref{sec:pok-rsd2})} & Fast  & 0.1 kB  & 12.6 kB & \multirow{2}{*}{5-round} & \multirow{2}{*}{$\IRSD$} \\ \cline{2-4}
                                                           & Short & 0.1 kB  & 10.2 kB &                          & \\ \cline{2-6}

                                                           & Fast  & 0.5 kB  & 8.4 kB & \multirow{2}{*}{5-round} & \multirow{2}{*}{$\IRSL$} \\ \cline{2-4}
                                                           & Short & 0.5 kB  & 6.1 kB &                          & \\ \hline

  \end{tabular}
  \vspace{0.5\baselineskip}
  \caption{Signatures based on $\RSD$ for $\lambda = 128$ (sorted by decreasing size)} \label{table:rsd2}
}}}
\end{center}
\end{table}

%\subsection{Signatures based on the $\MQ$ problem} \label{sec:param-mq}

%\input{parameters-mq}

%--------------------------------------------------------------------%
\section{Conclusion} \label{sec:conclusion}
%--------------------------------------------------------------------%

In this paper, we have introduced a new approach to design PoK along with its associated amplification technique.
Using this new paradigm, we have provided new post-quantum signatures related to the $\PKP$, $\SD$ and $\RSD$ problems.
Our signature related to the $\PKP$ problem features a (public key + signature) size ranging from 9kB to 10kB which is up to $45$\% shorter than existing ones.
Our signature related to the $\SD$ problem features a (public key + signature) size ranging from 15kB to 17kB which outperforms existing constructions such as Wave \cite{wave19} and LESS \cite{less-fm} but is outperformed by \cite{FJR22}.
Our signature related to the $\RSL$ problem has a (public key + signature) size ranging from 7kB to 9kB which outperforms Durandal \cite{Durandal}.
One should nonetheless note that Wave and Durandal have smaller signature sizes (but bigger public key sizes) than our schemes.
These constructions are interesting as they are also competitive with SPHINCS+ \cite{sphincs} that have been recently selected during the NIST Standardization Process.
While the MPC-in-the-head approach have opened the way to several trade-offs between signature size and performances, our work extend these possibilities even more by leveraging structured versions of the considered hard problems.
These new trade-offs are significant as they can lead to shorter signatures as demonstrated in this work.
Future work will include applying our new approach to other hard problems such as the $\MQ$ problem and $\SD$ over $\Fq$ one (see preliminary versions in Appendices~\ref{app:extra1} and \ref{app:extra2}).

%--------------------------------------------------------------------%
\bibliographystyle{alpha}
\bibliography{ref}

\newcommand{\etalchar}[1]{$^{#1}$}
\begin{thebibliography}{BGKM22}

\bibitem[ABG{\etalchar{+}}19]{Durandal}
Nicolas Aragon, Olivier Blazy, Philippe Gaborit, Adrien Hauteville, and Gilles
  Z{\'e}mor.
\newblock Durandal: a rank metric based signature scheme.
\newblock In {\em Annual International Conference on the Theory and
  Applications of Cryptographic Techniques (EUROCRYPT)}, 2019.

\bibitem[AFK21]{EPRINT:AttFehKlo21}
Thomas Attema, Serge Fehr, and Michael Klooß.
\newblock {Fiat-Shamir Transformation of Multi-Round Interactive Proofs}.
\newblock Cryptology ePrint Archive, Report 2021/1377, 2021.

\bibitem[AGS11]{AGS11}
Carlos {Aguilar Melchor}, Philippe Gaborit, and Julien {Schrek}.
\newblock A new zero-knowledge code based identification scheme with reduced
  communication.
\newblock In {\em IEEE Information Theory Workshop}, 2011.

\bibitem[BB21]{BB21}
Magali Bardet and Pierre Briaud.
\newblock {An algebraic approach to the Rank Support Learning problem}.
\newblock In {\em International Conference on Post-Quantum Cryptography
  (PQCrypto)}, 2021.

\bibitem[BBBG21]{ISIT21}
Slim Bettaieb, Loïc Bidoux, Olivier Blazy, and Philippe Gaborit.
\newblock {Zero-Knowledge Reparation of the Véron and AGS Code-based
  Identification Schemes}.
\newblock In {\em IEEE International Symposium on Information Theory (ISIT)},
  2021.

\bibitem[BBC{\etalchar{+}}20]{BBC20}
Magali Bardet, Maxime Bros, Daniel Cabarcas, Philippe Gaborit, Ray Perlner,
  Daniel Smith-Tone, Jean-Pierre Tillich, and Javier Verbel.
\newblock {Improvements of Algebraic Attacks for solving the Rank Decoding and
  MinRank problems}.
\newblock In {\em International Conference on the Theory and Application of
  Cryptology and Information Security (ASIACRYPT)}, 2020.

\bibitem[BBPS21]{less-fm}
Alessandro Barenghi, Jean-Fran{\c{c}}ois Biasse, Edoardo Persichetti, and Paolo
  Santini.
\newblock {LESS-FM: Fine-tuning Signatures from a Code-based Cryptographic
  Group Action}.
\newblock In {\em International Workshop on Post-Quantum Cryptography
  (PQCrypto)}, 2021.

\bibitem[BCG{\etalchar{+}}19]{RankAGS19}
Emanuele Bellini, Florian Caullery, Philippe Gaborit, Marc Manzano, and Victor
  Mateu.
\newblock {Improved V\'eron Identification and Signature Schemes in the rank
  metric}.
\newblock In {\em IEEE International Symposium on Information Theory (ISIT)},
  2019.

\bibitem[Beu20]{Beullens20}
Ward Beullens.
\newblock {Sigma Protocols for MQ, PKP and SIS, and Fishy Signature Schemes.}
\newblock {\em International Conference on the Theory and Applications of
  Cryptographic Techniques (EUROCRYPT)}, 2020.

\bibitem[BGKM22]{BGKM22}
Lo{\"\i}c Bidoux, Philippe Gaborit, Mukul Kulkarni, and Victor Mateu.
\newblock {Code-based Signatures from New Proofs of Knowledge for the Syndrome
  Decoding Problem}.
\newblock {\em arXiv preprint arXiv:2201.05403}, 2022.

\bibitem[BGKS21]{BGKS21}
Lo{\"\i}c Bidoux, Philippe Gaborit, Mukul Kulkarni, and Nicolas Sendrier.
\newblock {Quasi-Cyclic Stern Proof of Knowledge}.
\newblock {\em arXiv preprint arXiv:2110.05005}, 2021.

\bibitem[BHK{\etalchar{+}}19]{sphincs}
Daniel~J Bernstein, Andreas H{\"u}lsing, Stefan K{\"o}lbl, Ruben Niederhagen,
  Joost Rijneveld, and Peter Schwabe.
\newblock The sphincs+ signature framework.
\newblock In {\em {ACM SIGSAC Conference on Computer and Communications
  Security}}, 2019.

\bibitem[BJMM12]{BJMM12}
Anja Becker, Antoine Joux, Alexander May, and Alexander Meurer.
\newblock Decoding random binary linear codes in {$2^{n/20}$}: How {$1+1=0$}
  improves information set decoding.
\newblock In {\em International Conference on the Theory and Applications of
  Cryptographic Techniques (EUROCRYPT)}, 2012.

\bibitem[CDG{\etalchar{+}}20]{picnic}
Melissa Chase, David Derler, Steven Goldfeder, Daniel Kales, Johathan Katz,
  Vladimir Kolesnikov, Claudio Orlandi, Sebastian Ramacher, Christian
  Rechberger, Daniel Slamanig, Wang Xiao, and Greg Zaverucha.
\newblock {The Picnic Signature Algorithm}.
\newblock {\em NIST Post-Quantum Cryptography Standardization Project (Round
  3), \url{https://microsoft.github.io/Picnic/}}, 2020.

\bibitem[CHR{\etalchar{+}}20]{MQDSS}
Ming-Shing Chen, Andreas H{\"u}lsing, Joost Rijneveld, Simona Samardjiska, and
  Peter Schwabe.
\newblock {MQDSS specifications}.
\newblock {\em NIST Post-Quantum Cryptography Standardization Project (Round
  2), \url{https://mqdss.org/}}, 2020.

\bibitem[CVE11]{CVE11}
Pierre-Louis Cayrel, Pascal V{\'e}ron, and Sidi~Mohamed ElYousfi.
\newblock {A Zero-Knowledge Identification Scheme Based on the q-ary Syndrome
  Decoding Problem}.
\newblock In {\em Selected Areas in Cryptography (SAC)}, 2011.

\bibitem[DAST19]{wave19}
Thomas Debris-Alazard, Nicolas Sendrier, and Jean-Pierre Tillich.
\newblock Wave: A new family of trapdoor one-way preimage sampleable functions
  based on codes.
\newblock In {\em International Conference on the Theory and Application of
  Cryptology and Information Security (ASIACRYPT)}, 2019.

\bibitem[DAT18]{DT18}
Thomas Debris-Alazard and Jean-Pierre Tillich.
\newblock {Two attacks on rank metric code-based schemes: RankSign and an IBE
  scheme}.
\newblock In {\em International Conference on the Theory and Application of
  Cryptology and Information Security (ASIACRYPT)}, 2018.

\bibitem[FJR21]{FJR21}
Thibauld Feneuil, Antoine Joux, and Matthieu Rivain.
\newblock {Shared Permutation for Syndrome Decoding: New Zero-Knowledge
  Protocol and Code-Based Signature}.
\newblock {\em Cryptology ePrint Archive, Report 2021/1576}, 2021.

\bibitem[FJR22]{FJR22}
Thibauld Feneuil, Antoine Joux, and Matthieu Rivain.
\newblock Syndrome decoding in the head: Shorter signatures from zero-knowledge
  proofs.
\newblock {\em Cryptology ePrint Archive, Report 2022/188}, 2022.

\bibitem[FS86]{FS}
Amos Fiat and Adi Shamir.
\newblock How to prove yourself: Practical solutions to identification and
  signature problems.
\newblock In {\em Annual International Cryptology Conference (CRYPTO)}, 1986.

\bibitem[GHPT17]{GHPT17}
Philippe Gaborit, Adrien Hauteville, Duong~Hieu Phan, and Jean-Pierre Tillich.
\newblock {Identity-Based Encryption from Codes with Rank Metric}.
\newblock In {\em Annual International Cryptology Conference (CRYPTO)}, 2017.

\bibitem[GPS21]{GPS21}
Shay Gueron, Edoardo Persichetti, and Paolo Santini.
\newblock {Designing a Practical Code-based Signature Scheme from
  Zero-Knowledge Proofs with Trusted Setup}.
\newblock {\em Cryptology ePrint Archive, Report 2021/1020}, 2021.

\bibitem[GRS15]{RSD-attack}
Philippe Gaborit, Olivier Ruatta, and Julien Schrek.
\newblock {On the complexity of the Rank Syndrome Decoding problem}.
\newblock {\em IEEE Transactions on Information Theory}, 2015.

\bibitem[GZ08]{GZ08}
Philippe Gaborit and Gilles Zemor.
\newblock Asymptotic improvement of the gilbert-varshamov bound for linear
  codes.
\newblock {\em IEEE Transactions on Information Theory}, 54(9):3865--3872,
  2008.

\bibitem[IKOS07]{IKOS07}
Yuval Ishai, Eyal Kushilevitz, Rafail Ostrovsky, and Amit Sahai.
\newblock {Zero-Knowledge from Secure Multiparty Computation}.
\newblock In {\em Proceedings of the 39th annual ACM symposium on Theory of
  computing (STOC)}, 2007.

\bibitem[KKW18]{KKW18}
Jonathan Katz, Vladimir Kolesnikov, and Xiao Wang.
\newblock {Improved Non-Interactive Zero Knowledge with Applications to
  Post-Quantum Signatures}.
\newblock In {\em Proceedings of the 2018 ACM Conference on Computer and
  Communications Security (CCS)}, 2018.

\bibitem[KMRP19]{PKP-complexity}
Eliane Koussa, Gilles Macario-Rat, and Jacques Patarin.
\newblock {On the complexity of the Permuted Kernel Problem}.
\newblock {\em Cryptology ePrint Archive, Report 2019/412}, 2019.

\bibitem[KZ20]{KZ20}
Daniel Kales and Greg Zaverucha.
\newblock {An Attack on Some Signature Schemes Constructed From Five-Pass
  Identification Schemes}.
\newblock In {\em International Conference on Cryptology and Network Security
  (CANS)}, 2020.

\bibitem[PS96]{DPJS96}
David Pointcheval and Jacques Stern.
\newblock Security proofs for signature schemes.
\newblock In {\em International Conference on the Theory and Applications of
  Cryptographic Techniques (EUROCRYPT)}, 1996.

\bibitem[Sen11]{Sen11}
Nicolas Sendrier.
\newblock Decoding one out of many.
\newblock In {\em International Conference on Post-Quantum Cryptography
  (PQCrypto)}, 2011.

\bibitem[Sha89]{PKP}
Adi Shamir.
\newblock An efficient identification scheme based on permuted kernels.
\newblock In {\em Annual International Cryptology Conference (CRYPTO)}.
  Springer, 1989.

\bibitem[SSH11]{Sakumoto11}
Koichi Sakumoto, Taizo Shirai, and Harunaga Hiwatari.
\newblock {Public-Key Identification Schemes Based on Multivariate Quadratic
  Polynomials}.
\newblock In {\em Annual International Cryptology Conference (CRYPTO)}.
  Springer, 2011.

\bibitem[Ste93]{Stern93}
Jacques Stern.
\newblock A new identification scheme based on syndrome decoding.
\newblock In {\em Annual International Cryptology Conference (CRYPTO)}, 1993.

\bibitem[Unr15]{Unruh15}
Dominique Unruh.
\newblock {Non-Interactive Zero-Knowledge Proofs in the Quantum Random Oracle
  Model}.
\newblock In {\em Annual International Conference on the Theory and
  Applications of Cryptographic Techniques (EUROCRYPT)}, 2015.

\bibitem[V{\'e}r97]{Veron97}
Pascal V{\'e}ron.
\newblock Improved identification schemes based on error-correcting codes.
\newblock {\em Applicable Algebra in Engineering, Communication and Computing},
  1997.

\bibitem[Wan22]{Wang22}
William Wang.
\newblock {Shorter Signatures from MQ}.
\newblock {\em Cryptology ePrint Archive, Report 2022/344}, 2022.

\end{thebibliography}
%--------------------------------------------------------------------%
\appendix

\section{Proof of Theorem~\ref{thm:pok-sd1}} \label{app:pok-sd1}

\textbf{Theorem 1.}
\emph{If the hash function used is collision-resistant and if the commitment scheme used is binding and hiding, then the protocol depicted in Figure~\ref{fig:pok4-h} is an honest-verifier zero-knowledge PoK with Helper for the $\SD$ problem over $\Ft$ with soundness error $1/N$.}

\begin{proof} We prove the correctness, special soundness and special honest-verifier zero-knowledge properties below.

\vspace{\baselineskip}
\noindent \textbf{Correctness.}
The correctness follows from the protocol description once the cut-and-choose with meet in the middle property $\bar{\bm{s}}_{\alpha} = \bar{\bm{t}}_{\alpha} + \bm{z}_4$ has been verified.
  From $\bm{s}_0 = \bm{u} + \bm{x}$ and $\bm{s}_i = \pi_i[\bm{s}_{i - 1}] + \bm{v}_i$ for all $i \in \intoneto{\alpha}$, one can see that $\bar{\bm{s}}_{\alpha} = \pi_{\alpha} \circ \cdots \circ \pi_{1}[\bm{u} + \bm{x}] + \bm{v}_{\alpha} + \sum\nolimits_{i \in \intoneto{\alpha - 1}} \pi_\alpha \circ \cdots \circ \pi_{i + 1}[\bm{v}_i]$.
  In addition, from $\bar{\bm{t}}_N = \pi[\bm{u}] + \bm{v}$, and $\bar{\bm{t}}_{i-1} = \pi^{-1}_i[\bar{\bm{t}}_i - \bm{v}_i]$ for all $i \in \{N, \ldots, \alpha+1 \}$, one can see that $\bar{\bm{t}}_{\alpha} = \pi_{\alpha} \circ \cdots \circ \pi_{1}[\bm{u}] + \bm{v}_{\alpha} + \sum\nolimits_{i \in \intoneto{\alpha - 1}} \pi_\alpha \circ \cdots \circ \pi_{i + 1}[\bm{v}_i]$.
As $\bm{z}_4 = \pi_{\alpha} \circ \dots \circ \pi_1[\bm{x}]$, one can conclude that $\bar{\bm{s}}_{\alpha} = \bar{\bm{t}}_{\alpha} + \bm{z}_4$.

\vspace{\baselineskip}
\noindent \textbf{Special soundness.} 
  To prove the special soundness, one need to build an efficient knowledge extractor $\Ext$ which returns a solution of the $\SD$ instance defined by $(\bm{H}, \bm{y})$ given two valid transcripts $(\bm{H}, \bm{y}, \com_1, \com_2, \alpha, \allowbreak \rsp)$ and $(\bm{H}, \bm{y}, \com_1, \allowbreak \com_2, \alpha', \rsp')$ with $\alpha \neq \alpha'$  where $\com_1 = \setup(\theta, \xi)$ for some random seeds ($\theta, \xi)$.
The knowledge extractor $\Ext$ computes the solution as:

\vspace{0.5\baselineskip} 
\pseudocode{%
  \tx{1. Compute } (\pi_i)_{i \in \intoneto{n}} \tx{ from } z_2 \tx{ and } z_2' \\
  \tx{2. Output } (\pi_{1}^{-1} \circ \cdots \circ \pi_{\alpha}^{-1}[\bm{z}_4])
}
\vspace{0.5\baselineskip} 

\noindent We now show that the output is a solution to the given $\SD$ problem.
One can compute $(\bar{\pi}_i, \bar{\bm{v}}_i)_{i \in \intoneto{N}}$ from $z_2$ and $z_2'$.
From the binding property of the commitments $(\com_{1,i})_{i \in \intoneto{N}}$, one has $(\pi_i, \bm{v}_i)_{i \in \intoneto{N}} = (\bar{\pi}_i, \bar{\bm{v}}_i)_{i \in \intoneto{N}}$.
From the binding property of commitment $\com_1$, one has $\bm{H}(\bm{z}_1 - \bm{u}) = \bm{y}$ and $\bar{\bm{t}}_N = \pi[\bm{u}] + \bm{v}$.
Using $\bar{\bm{t}}_N$ and $(\pi_i, \bm{v}_i)_{i \in \intoneto{N}}$, one has $\bar{\bm{t}}_{\alpha} = \pi_{\alpha} \circ \cdots \circ \pi_{1}[\bm{u}] + \bm{v}_{\alpha} + \sum\nolimits_{i \in \intoneto{\alpha - 1}} \pi_\alpha \circ \cdots \circ \pi_{i + 1}[\bm{v}_i]$.
From the binding property of commitment $\com_2$, one has $\bar{\bm{s}}_0 = \bar{\bm{s}}_0' = \bm{z}_1$.
In addition, one has $\bar{\bm{s}}_i = \bar{\pi}_i[\bar{\bm{s}}_{i - 1}] + \bar{\bm{v}}_i$ for all $i \in \intoneto{N} \setminus \alpha$ as well as $\bar{\bm{s}}_i' = \bar{\pi}_i[\bar{\bm{s}}_{i - 1}'] + \bar{\bm{v}}_i$ for all $i \in \intoneto{N} \setminus \alpha'$.
  Using the binding property of commitment $\com_2$ once again, one can deduce that $\bar{\bm{s}}_i = \bar{\pi}_i[\bar{\bm{s}}_{i - 1}] + \bar{\bm{v}}_i$ for all $i \in \intoneto{N}$ hence $\bar{\bm{s}}_{\alpha} = \pi_{\alpha} \circ \cdots \circ \pi_{1}[\bm{z}_1] + \bm{v}_{\alpha} + \sum\nolimits_{i \in \intoneto{\alpha - 1}} \pi_{\alpha} \circ \cdots \circ \pi_{i + 1}[\bm{v}_i]$.
From the binding property of commitment $\com_2$, one has $\bar{\bm{s}}_{\alpha} = \bar{\bm{t}}_{\alpha} + \bm{z}_4$ hence $\bm{z}_1 - \bm{u} = \pi^{-1}_{1} \circ \cdots \circ \pi^{-1}_{\alpha}[\bm{z}_4]$.
As a consequence, one has $\bm{H}(\pi^{-1}_{1} \circ \cdots \circ \pi^{-1}_{\alpha}[\bm{z}_4]) = \bm{y}$ along with $\hw{\bm{z}_4} = \omega$ thus $\pi^{-1}_{1} \circ \cdots \circ \pi^{-1}_{\alpha}[\bm{z}_4]$ is a solution of the considered $\SD$ problem instance.

\vspace{\baselineskip}
\noindent \textbf{Special Honest-Verifier Zero-Knowledge.} 
We start by explaining why valid transcripts do not leak anything on the secret $\bm{x}$.
  A valid transcript contains $(\bm{u} + \bm{x}, \, (\pi_i, \bm{v}_i)_{i \in \intoneto{N} \setminus \alpha}, \, \pi[\bm{u}] + \bm{v}, \, \pi_{\alpha} \circ \cdots \circ \pi_{1}[\bm{x}], \com_{1, \alpha})$ namely the secret $\bm{x}$ is masked either by a random value $\bm{u}$ or by a random permutation $\pi_{\alpha}$.
  The main difficulty concerns the permutation $\pi_\alpha$ as the protocol requires $\pi_{\alpha} \circ \cdots \circ \pi_1[\bm{u} + \bm{x}]$ to be computed while both $(\bm{u} + \bm{x})$ and $(\pi_i)_{i \in \intoneto{\alpha-1}}$ are known.
  To overcome this issue, the protocol actually computes $\pi_{\alpha} \circ \cdots \circ \pi_1[\bm{u} + \bm{x}] + \bm{v}_{\alpha} + \sum\nolimits_{i \in \intoneto{\alpha - 1}} \pi_\alpha \circ \cdots \circ \pi_{i + 1}[\bm{v}_i]$ for some random value $\bm{v}_{\alpha}$ hence does not leak anything on $\pi_{\alpha}$.
  In addition, if the commitment used is hiding, $\com_{1, \alpha}$ does not leak anything on $\pi_{\alpha}$ nor $\bm{v}_{\alpha}$.
  Formally, one can build a $\ppt$ simulator $\Sim$ that given the public values $(\bm{H}, \bm{y})$, random seeds $(\theta, \xi)$ and a random challenge $\alpha$ outputs a transcript $(\bm{H}, \bm{y}, \com_1, \com_2, \alpha, \rsp)$ such that $\com_1 = \setup(\theta, \xi)$ that is indistinguishable from the transcript of honest executions of the protocol:

\vspace{0.5\baselineskip}
\pseudocode{%
  \tx{1. Compute } (\pi_i, \bm{v}_i)_{i \in \intoneto{N}} \tx{ and } \bm{u} \tx{ from } (\theta, \xi) \\
  \tx{2. Compute } \bm{\tilde{x}}_1 \tx{ such that } \bm{H} \bm{\tilde{x}}_1 = \bm{\bm{y}} \tx{ and } \bm{\tilde{x}}_2 \sampler \swset{\omega}{\Ftn} \\
  \tx{3. Compute } \bm{\tilde{s}}_0 = \bm{u} + \bm{\tilde{x}}_1 \tx { and } \bm{\tilde{s}}_i = \pi_i[\tilde{\bm{s}}_{i-1}] + \bm{v}_i \tx{ for all } i \in \intoneto{\alpha-1} \\
  \tx{4. Compute } \bm{\tilde{s}}_{\alpha} = \pi_{\alpha} \circ \cdots \circ \pi_1[\bm{u} + \bm{\tilde{x}}_2] + \bm{v}_{\alpha} + \sum\nolimits_{i \in \intoneto{\alpha - 1}} \pi_\alpha \circ \cdots \circ \pi_{i + 1}[\bm{v}_i] \\
  \tx{5. Compute } \bm{\tilde{s}}_i = \pi_i[\tilde{\bm{s}}_{i-1}] + \bm{v}_i \tx{ for all } i \in [\alpha+1,N] \\
  \tx{6. Compute } \tilde{\com}_2 = \commitb{\bm{u} + \bm{\tilde{x}}_1 \, || \, (\bm{\tilde{s}}_i)_{i \in \intoneto{N}}} \\
  \tx{7. Compute } \bm{\tilde{z}}_1 = \bm{u} + \bm{\tilde{x}}_1, ~ z_2 = (\theta_{i})_{i \in \intoneto{N} \setminus \alpha}, ~ z_3 = \xi, ~ \bm{z}_4 = \pi_{\alpha} \circ \dots \circ \pi_1[\bm{\tilde{x}}_2] \\
  \tx{8. Compute  } \tilde{\rsp} = (\bm{\tilde{z}}_1, z_2, z_3, \bm{\tilde{z}}_4, \com_{1, \alpha}) \tx{ and output } (\bm{H}, \bm{y}, \com_1, \tilde{\com}_2, \alpha, \tilde{\rsp}) 
}
\vspace{0.5\baselineskip}

\noindent The transcript generated by the simulator $\Sim$ is $(\bm{H}, \bm{y}, \com_1, \tilde{\com}_2, \alpha, \tilde{\rsp})$ where $\com_1 \samplen \setup(\theta, \xi)$.
Since $\bm{\tilde{x}}_1$ and $\bm{x}$ are masked by a random mask $\bm{u}$ unknown to the verifier, $\bm{\tilde{z}}_1$ and $\bm{z}_1$ are indistinguishable.
Similarly, since $\bm{\tilde{x}}_2$ and $\bm{x}$ have the same Hamming weight and are masked by a random permutation $\pi_{\alpha}$ unknown to the verifier, $\bm{\tilde{z}}_4$ and $\bm{z}_4$ are indistinguishable.
As $\bm{\tilde{z}}_1$ and $\bm{z}_1$ are indistinguishable, $\bm{\tilde{s}}_i$ and $\bm{s}_i$ are also indistinguishable for all $i \in \intoneto{\alpha - 1}$.
Since $\bm{\tilde{s}}_\alpha$ and $\bm{s}_\alpha$ both contains a random mask $\bm{v}_\alpha$ unknown to the verifier, they are indistinguishable.
As $\bm{\tilde{s}}_\alpha$ and $\bm{s}_\alpha$ are indistinguishable, so do $\bm{\tilde{s}}_i$ and $\bm{s}_i$ for all $i \in [\alpha+1, N]$.
Finally, $z_2$ and $z_3$ are identical in both cases and $\com_{1, \alpha}$ does not leak anything if the commitment is hiding.
As a consequence, $(\tilde{\rsp}, \tilde{\com}_2)$ in the simulation and $(\rsp, \com_2)$ in the real execution are indistinguishable.
Finally, $\Sim$ runs in polynomial time which completes the proof.
\end{proof}

\section{Proof of Theorem~\ref{thm:pok-pkp2}} \label{app:pok-pkp2}

\vspace{\baselineskip}
\noindent \textbf{Theorem 2.}
\emph{If the hash function used is collision-resistant and if the commitment scheme used is binding and hiding, then the protocol depicted in Figure~\ref{fig:pkp} is an honest-verifier zero-knowledge PoK for the $\IPKP$ problem with soundness error equal to $\frac{1}{N} + \frac{N - 1}{N \cdot (q - 1)}$.}

\begin{proof} 
We prove the correctness, special soundness and special honest-verifier zero-knowledge properties below.

\vspace{\baselineskip}
\noindent \textbf{Correctness.}
The correctness follows from the protocol description once it is observed that $\bm{s}_N = \pi[\kappa \cdot \bm{x}] + \bm{v}$ which implies that $\bm{H} \bm{s}_N - \kappa \cdot \bm{y} = \bm{H} \pi[\kappa \cdot \bm{x}] + \bm{H} \bm{v} - \kappa \cdot \bm{y} = \bm{H} \bm{v}$.

\vspace{\baselineskip}
\noindent \textbf{$(q-1,N)$-special soundness.}
  To prove the $(q-1,N)$-special soundness, one need to build an efficient knowledge extractor $\Ext$ which returns a solution of the $\IPKP$ instance defined by $(\bm{H}, \bm{x}, \bm{y})$ with high probability given a $(q-1,N)$-tree of accepting transcripts.
  One only need a subset of the tree to complete the proof namely the four leafs corresponding to challenges $(\kappa, \alpha_1), (\kappa, \alpha_2), (\kappa', \alpha_1)$ and $(\kappa', \alpha_2)$ where $\kappa \neq \kappa'$ and $\alpha_1 \neq \alpha_2$. 
The knowledge extractor $\Ext$ computes the solution as:

\vspace{0.5\baselineskip} 
\pseudocode{%
  \tx{1. Compute } (\bar{\pi}_i)_{i \in \intoneto{n}} \tx{ from } z_2^{(\kappa, \alpha_1)} \tx{ and } z_2^{(\kappa, \alpha_2)} \\
  \tx{2. Compute } \bar{\pi} = \bar{\pi}_N \circ \cdots \circ \bar{\pi}_1 \\
  \tx{3. Output } \bar{\pi} 
}
\vspace{0.5\baselineskip} 

\noindent One can compute $(\bar{\pi}^{(\kappa)}_i, \bar{\bm{v}}^{(\kappa)}_i)_{i \in \intoneto{N}}$ and $(\bar{\pi}^{(\kappa')}_i, \bar{\bm{v}}^{(\kappa')}_i)_{i \in \intoneto{N}}$ from $\big( z_2^{(\kappa, \alpha_i)} \big)_{i \in [1,2]}$ and  $\big( z_2^{(\kappa', \alpha_i)} \big)_{i \in [1,2]}$ respectively.
  From the binding property of the commitments $(\com_{1,i})_{i \in \intoneto{N}}$, one has $(\bar{\pi}_i, \bar{\bm{v}}_i)_{i \in \intoneto{N}} = (\bar{\pi}^{(\kappa)}_i, \bar{\bm{v}}^{(\kappa)}_i)_{i \in \intoneto{N}} = (\bar{\pi}^{(\kappa')}_i, \bar{\bm{v}}^{(\kappa')}_i)_{i \in \intoneto{N}}$.
By construction, one has $\bar{\bm{s}}^{(\kappa, \alpha_1)}_0 = \bar{\bm{s}}^{(\kappa, \alpha_2)}_0 = \kappa \cdot \bm{x}$.
In addition, one has $\bar{\bm{s}}^{(\kappa, \alpha_1)}_i = \bar{\pi}_i[\bar{\bm{s}}^{(\kappa, \alpha_1)}_{i - 1}] + \bar{\bm{v}}_i$ for all $i \in \intoneto{N} \setminus \alpha_1$ as well as $\bar{\bm{s}}^{(\kappa, \alpha_2)}_i = \bar{\pi}_i[\bar{\bm{s}}^{(\kappa, \alpha_2)}_{i - 1}] + \bar{\bm{v}}_i$ for all $i \in \intoneto{N} \setminus \alpha_2$.
From the binding property of commitment $\com_2$, one can deduce that $\bar{\bm{s}}^{(\kappa)}_i = \bar{\pi}_i[\bar{\bm{s}}^{(\kappa)}_{i - 1}] + \bar{\bm{v}}_i$ for all $i \in \intoneto{N}$ hence $\bar{\bm{s}}^{(\kappa)}_N = \bar{\pi}[\kappa \cdot \bm{x}] + \bar{\bm{v}}$.
Following a similar argument, one also has $\bar{\bm{s}}^{(\kappa')}_N = \bar{\pi}[\kappa' \cdot \bm{x}] + \bar{\bm{v}}$.
From the binding property of commitment $\com_1$, one has $\bm{H} \bar{\bm{s}}^{(\kappa)}_N - \kappa \cdot \bm{y} = \bm{H} \bar{\bm{s}}^{(\kappa')}_N - \kappa' \cdot \bm{y}$.
It follows that $\bm{H}(\bar{\pi}[\kappa \cdot \bm{x}] + \bar{\bm{v}}) - \kappa \cdot \bm{y} = \bm{H}(\bar{\pi}[\kappa' \cdot \bm{x}] + \bar{\bm{v}}) - \kappa' \cdot \bm{y}$ hence $(\kappa - \kappa') \cdot \bm{H} \bar{\pi}[\bm{x}] = (\kappa - \kappa') \cdot \bm{y}$.
This implies that $\bm{H} \bar{\pi}[\bm{x}] = \bm{y}$ thus $\bar{\pi}$ is a solution of the considered $\IPKP$ problem.

\vspace{\baselineskip}
\noindent \textbf{Special Honest-Verifier Zero-Knowledge.} 
We start by explaining why valid transcripts do not leak anything on the secret $\pi$.
A valid transcript contains $(\bm{s}_{\alpha}, \, (\pi_i, \bm{v}_i)_{i \in \intoneto{N} \setminus \alpha}, \, \com_{1, \alpha})$ where the secret $\pi$ is hiden by the unknown permutation $\pi_{\alpha}$.
In our protocol, one need to compute $\pi[\bm{x}]$ without leaking anything on the secret $\pi$.
To overcome this issue, the protocol actually computes $\pi[\bm{x}] + \bm{v}$ for some value $\bm{v}$ that is masked by the unknown random value $\bm{v}_\alpha$.
In addition, if the commitment used is hiding, $\com_{1, \alpha}$ does not leak anything on $\pi_{\alpha}$ nor $\bm{v}_{\alpha}$.
Formally, one can build a $\ppt$ simulator $\Sim$ that given the public values $(\bm{H}, \bm{x}, \bm{y})$, random challenges $(\kappa, \alpha)$ outputs a transcript $(\bm{H}, \bm{x}, \bm{y}, \com_1, \kappa, \com_2, \alpha, \rsp)$ that is indistinguishable from the transcript of honest executions of the protocol:

\vspace{0.5\baselineskip}
\pseudocode{%
  \tx{1. Compute } (\pi_i, \bm{v}_i, \tilde{\com}_{1,i}) \tx{ as in the real protocol except for } \tilde{\pi}_1 \sampler S_n \\
  \tx{2. Compute } \tilde{\pi} = \pi_N \circ \cdots \tilde{\pi}_1 \\
  \tx{3. Compute } \bm{v} \tx{ and } \tilde{\com_1} \tx{ as in the real protocol} \\
  \tx{4. Compute } \bm{\tilde{x}} \tx{ such that } \bm{H} \bm{\tilde{x}} = \kappa \cdot \bm{\bm{y}} \\
  \tx{5. Compute } \bm{s}_0 = \kappa \cdot \bm{x} \tx { and } \bm{\tilde{s}}_i = \pi_i[\bm{\tilde{s}}_{i-1}] + \bm{v}_i \tx{ for all } i \in \intoneto{\alpha-1} \\
  \tx{6. Compute } \bm{\tilde{s}}_{\alpha} = \pi_{\alpha}[\bm{\tilde{s}}_{\alpha - 1}] + \bm{v}_{\alpha} + \pi^{-1}_{\alpha + 1} \circ \cdots \circ \pi^{-1}_{N}[\bm{\tilde{x}} - \pi[\kappa \cdot \bm{x}]] \\
  \tx{7. Compute } \bm{\tilde{s}}_i = \pi_i[\bm{\tilde{s}}_{i-1}] + \bm{v}_i \tx{ for all } i \in [\alpha+1,N] \\
  \tx{8. Compute } \tilde{\com}_2 = \commitb{(\bm{\tilde{s}}_i)_{i \in \intoneto{N}}} \tx{ and } \bm{\tilde{z}}_1 = \bm{\tilde{s}}_{_\alpha} \\
  \tx{9. Compute } \tilde{z}_2 = \tilde{\pi}_1 \, || \, (\theta_{i})_{i \in \intoneto{N} \setminus \alpha} \text{ if } \alpha \neq 1 \tx{ or } \bar{z}_2 = (\theta_{i})_{i \in \intoneto{N} \setminus \alpha} \tx{ otherwise} \\
  \tx{10. Compute  } \tilde{\rsp} = (\bm{\tilde{z}}_1, \tilde{z}_2, \tilde{\com}_{1, \alpha}) \tx{ and output } (\bm{H}, \bm{x}, \bm{y}, \tilde{\com}_1, \kappa, \tilde{\com}_2, \alpha, \tilde{\rsp}) 
}
\vspace{0.5\baselineskip}

\noindent The transcript generated by the simulator $\Sim$ is $(\bm{H}, \bm{x}, \bm{y}, \tilde{\com}_1, \kappa, \tilde{\com}_2, \alpha, \tilde{\rsp})$.
Since $\bm{\tilde{s}}_{\alpha}$ (in the simulation) and $\bm{s}_{\alpha}$ (in the real world) are masked by a random mask $\bm{v}_{\alpha}$ unknown to the verifier, $\bm{\tilde{z}}_1$ and $\bm{z}_1$ are indistinguishable.
In addition, since $\tilde{\pi}_1$ is sampled uniformly at random in $\hperm{n}$, $\tilde{z}_2$ and $z_2$ are indistinguishable.
Finally, $\tilde{\com}_{1, \alpha}$ does not leak anything on $\pi_{\alpha}$ nor $\bm{v}_{\alpha}$ if the commitment is hiding.
As a consequence, $(\tilde{\com}_1, \tilde{\com}_2, \tilde{\rsp})$ (in the simulation) and $(\com_1, \com_2, \rsp)$ (in the real execution) are indistinguishable.
Finally, $\Sim$ runs in polynomial time which completes the proof.
\end{proof}

\section{Proof of Theorem~\ref{thm:pok-sd2-qc}} \label{app:pok-sd2-qc}

Similarly to what was done in \cite{AGS11}, we introduce the intermediary $\DiffSD$ problem (Definition~\ref{def:diffsd}) in order to prove the security of the protocol depicted in Figure~\ref{fig:pok4-qc}.
Its security (Theorem~\ref{thm:pok-sd2-qc}) relies of the $\DiffSD$ problem and is completed by a reduction from the $\QCSD$ problem to the $\DiffSD$ problem (Theorem~\ref{thm:reduction-diffsd}).
In our context, we consider $\QCSD$ instances with up to $M$ vectors (decoding one out of many setting) which means that the adversary has access to $Mk$ syndromes ($M$ given syndromes combined with $k$ possible shifts). 
In practice, one has to choose the $\QCSD$ parameters so that the PoK remains secure even taking into account both the number of given syndromes as well as the (small) security loss induced by the use of the $\DiffSD$ problem.

\vspace{0.25\baselineskip}
\begin{definition}[$\DiffSD$ problem] \label{def:diffsd}
  Let $(n=2k, k, w, M, \Delta)$ be positive integers, $\bm{H} \in \mathcal{QC}(\Ft^{(n - k) \times n})$ be a random parity-check matrix of a quasi-cyclic code of index $2$, $(\bm{x}_i)_{i \in \intoneto{M}} \in (\Ft^{n})^M$ be vectors such that $\hw{\bm{x}_i} = w$ and $(\bm{y}_i)_{i \in \intoneto{M}} \in (\Ft^{(n-k)})^M$ be vectors such that $\bm{H} \bm{x}_i^\top = \bm{y}_i^\top$. 
  Given $(\bm{H}, (\bm{y}_i)_{i \in \intoneto{M}})$, the Differential Syndrome Decoding problem $\DiffSD(n, k, w, M, \Delta)$ asks to find $(\bm{c}, (\bm{d}_{j}, \kappa_j, \allowbreak \mu_j)_{j \allowbreak \in \intoneto{\Delta}}) \in \Ft^{(n - k)} \times (\Ft^{n} \times \intoneto{k} \times \intoneto{M})^{\Delta}$ such that $\bm{H}\bm{d}_{j}^\top + \bm{c} =\bm{rot}_{\kappa_j}(\bm{y}_{\mu_j}^\top)$ and $\hw{\bm{d}_{j}} = w$ for each $j \in \intoneto{\Delta}$.

\end{definition}

\vspace{0.25\baselineskip}
\begin{theorem} \label{thm:reduction-diffsd}
  If there exists a $\ppt$ algorithm solving the $\DiffSD(n,k, \allowbreak w, M, \Delta)$ problem with probability $\epsilon_{\DiffSD}$, then there exists a $\ppt$ algorithm solving the $\QCSD(n, \allowbreak k, w, M)$ with probability $\epsilon_{\QCSD} \ge (1 - M \times p - (2^{(n - k)} - 2) \times p^{\Delta}) \cdot \epsilon_{\DiffSD}$ where $p = \frac{{n \choose \omega}}{2^{(n - k)}}$.
\end{theorem}

\noindent \emph{Sketch of Proof.}
We start by highlighting the main steps of the proof.
One should note that the $\DiffSD$ problem is constructed from a $\QCSD$ instance and as such always admit at least a solution namely the solution of the underlying $\QCSD$ instance.
Indeed, any solution to the $\DiffSD$ problem satisfying $\bm{c} = (0, \cdots, 0)$ can be transformed into a solution to the $\QCSD$ problem with similar inputs.
Hereafter, we study the probability that there exists solutions to the $\DiffSD$ problem for any possible value of $\bm{c}$.
To do so, we consider two cases depending on wheither $\bm{c}$ is stable by rotation or not.
The first case implies that either $\bm{c} = (0, \cdots, 0)$ or $\bm{c} = (1, \cdots, 1)$ while the second case encompasses every other possible value for $\bm{c}$.
We show that for correctly chosen values $n$, $k$, $w$ and $\Delta$, the probability that there exists solutions to the $\DiffSD$ problem satisfying $\bm{c} \ne (0, \cdots, 0)$ is small.
Such solutions can't be transformed into solutions to the $\QCSD$ problem hence induce a security loss in our reduction.

Given a $[n, k]$ quasi-cyclic code $\mathcal{C}$, we restrict our analysis (and our parameters choice) to the case where (i) $n$ is a primitive prime and (ii) the weight $\omega$ is lower than the Gilbert-Varshamov bound associated to $\mathcal{C}$ \textit{i.e.} the value for which the number of words of weight less or equal to $w$ corresponds to the number of syndromes.
Thus, given a syndrome $\bm{y}$, the probability $p$ that there exists a pre-image $\bm{x}$ of $\bm{y}$ such that $\bm{Hx}^\top = \bm{y}^\top$ and $\hw{\bm{x}} = \omega$ is $p = \binom{n}{w} / 2^{(n-k)}$ namely the number of possible words of weight $\omega$ divided by the number of syndromes.

Let $\mathcal{A}_{\DiffSD}$ be an algorithm that given inputs $(\bm{H}, (\bm{y}_i)_{i \in \intoneto{M}})$ generated following Definition~\ref{def:diffsd} outputs a solution $(\bm{c}, (\bm{d}_j, \kappa_j, \mu_j)_{j \in \intoneto{\Delta}})$ to the considered $\DiffSD$ instance.
Let $\mathcal{A}_{\QCSD}$ be an algorithm that given access to $\mathcal{A}_{\DiffSD}$ and inputs $(\bm{H}, (\bm{y}_i)_{i \in \intoneto{M}})$ corresponding to an instance of the $\QCSD$ problem in the decoding one out of many setting outputs a solution to this instance.
We denote by $\mathcal{A}_{\QCSD}(\bm{H}, (\bm{y}_i)_{i \in \intoneto{M}}) \allowbreak \ne \bot$ (respectively $\mathcal{A}_{\DiffSD}(\bm{H}, (\bm{y}_i)_{i \in \intoneto{M}}) \allowbreak \ne \bot$) the fact that $\mathcal{A}_{\QCSD}$ (respectively $\mathcal{A}_{\DiffSD}$) outputs a \emph{valid} solution to the $\QCSD$ (respectively $\DiffSD$) problem.

\vspace{0.5\baselineskip}
\noindent \underline{$\mathcal{A}_{\QCSD}(\bm{H}, (\bm{y}_i)_{i \in \intoneto{M}})$:} \\
1. Compute $(\bm{c}, (\bm{d}_j, \kappa_j, \mu_j)_{j \in \intoneto{\Delta}}) \leftarrow \mathcal{A}_{\DiffSD}(\bm{H}, (\bm{y}_i)_{i \in \intoneto{M}})$ \\
2. If $\bm{c} = (0, \cdots, 0)$, output $\bm{x} = \rot_{k - \kappa_1}(\bm{d}_1)$ \\
3. If $\bm{c} \neq (0, \cdots, 0)$, output $\bot$
\vspace{0.5\baselineskip}

Let $\mathtt{c_{0}}$ denote the event that the solution to the $\DiffSD$ problem is also the solution of the underlying $\QCSD$ instance.
One has $\epsilon_{\QCSD} = P[\mathcal{A}_{\QCSD}(\bm{H}, (\bm{y}_i)_{i \in \intoneto{M}}) \allowbreak \ne \bot] \ge P[\mathcal{A}_{\QCSD}(\bm{H}, (\bm{y}_i)_{i \in \intoneto{M}}) \allowbreak \ne \bot \, \cap \, \mathtt{c_{0}}] = P[\mathcal{A}_{\DiffSD}(\bm{H}, (\bm{y}_i)_{i \in \intoneto{M}}) \allowbreak \ne \bot \, \cap \, \mathtt{c_{0}}]$.
Let $\mathtt{c_{stable}}$ and $\mathtt{c_{unstable}}$ denote the events that the $\DiffSD$ problem admits another solution than the one of its underlying $\QCSD$ instance where $\bm{c}$ is stable (respectively unstable) by rotation.
One has $P[\mathcal{A}_{\DiffSD}(\bm{H}, (\bm{y}_i)_{i \in \intoneto{M}}) \ne \bot] = P[\mathcal{A}_{\DiffSD}(\bm{H}, (\bm{y}_i)_{i \in \intoneto{M}}) \ne \bot \cap \mathtt{c_{stable}}] + P[\mathcal{A}_{\DiffSD}(\bm{H}, (\bm{y}_i)_{i \in \intoneto{M}}) \ne \bot \cap \mathtt{c_{unstable}}]$.
We show bellow that if $\bm{c}$ is stable by rotation then $\bm{c} = (0, \cdots, 0)$ or $\bm{c} = (1, \cdots, 1)$.
Let $\mathtt{c_{1}}$ denote the event that the $\DiffSD$ problem admits another solution than the one of its underlying $\QCSD$ instance where $\bm{c} = (1, \cdots, 1)$. 
It follows that $P[\mathcal{A}_{\DiffSD}(\bm{H}, (\bm{y}_i)_{i \in \intoneto{M}}) \ne \bot \, \cap \, \mathtt{c_{0}}] = P[\mathcal{A}_{\DiffSD}(\bm{H}, (\bm{y}_i)_{i \in \intoneto{M}}) \ne \bot] - \allowbreak P[\mathcal{A}_{\DiffSD}(\bm{H}, \allowbreak (\bm{y}_i)_{i \in \intoneto{M}}) \ne \bot \, \cap \, \mathtt{c_{1}}] - P[\mathcal{A}_{\DiffSD}(\bm{H}, (\bm{y}_i)_{i \in \intoneto{M}}) \ne \bot \, \cap \, \mathtt{c_{unstable}}]$ hence $P[\mathcal{A}_{\DiffSD}(\bm{H}, (\bm{y}_i)_{i \in \intoneto{M}}) \ne \bot \, \cap \, \mathtt{c_{0}}]  = \epsilon_{\DiffSD} - P[\mathcal{A}_{\DiffSD}(\bm{H}, (\bm{y}_i)_{i \in \intoneto{M}}) \ne \bot \cap \mathtt{c_{1}}] - P[\mathcal{A}_{\DiffSD}(\bm{H}, (\bm{y}_i)_{i \in \intoneto{M}}) \allowbreak \ne \bot \cap \mathtt{c_{unstable}}] = (1 - P[\mathtt{c_1}] - P[\mathtt{c_{unstable}}]) \cdot \epsilon_{\DiffSD}$.
It follows that $\epsilon_{\QCSD} \ge (1 - P[\mathtt{c_1}] - P[\mathtt{c_{unstable}}]) \cdot \epsilon_{\DiffSD}$.

Working modulo $x^n-1$ and writing $\bm{c}$ as $c(x)=\sum_{i=0}^{n-1} c_ix^i$ being stable by rotation of order $j$ implies $x^jc(x)=c(x) \mod {x^n-1}$ hence $(x^j+1)c(x)=0 \mod {x^n-1}$. In our case where $2$ is primitive modulo $n$, one has $x^n-1=(x - 1)(1 + x + x^2 + \cdots + x^{n-1})$ where $(1 + x + x^2 + \cdots + x^{n-1})$ is an irreducible polynomial \cite{GZ08}.
Since $c(x)$ divides $x^n-1$ and since $(x - 1)$ is not compatible with $\bm{c}$ being stable by rotation, the only non zero possibility is $\bm{c} = (1, \cdots, 1)$.

Hereafter, we compute $P[\mathtt{c_1}]$ and $P[\mathtt{c_{unstable}}]$.
The probability that the $\DiffSD$ problem admits another solution than the one of its underlying $\QCSD$ instance where $\bm{c} = (1, \cdots, 1)$ is the same as the probability that the vector $\bm{y}_{\mu_j} - \bm{c}$ has a preimage by $\bm{H}$ of weight $w$ namely $p$.
As $M$ vectors $(\bm{y}_{\mu_j})_{j \in \intoneto{M}}$ can be considered, it follows that $P[\mathtt{c_1}] = M \times p = M \times {n \choose \omega}/2^{(n - k)}$.

In the case where $\bm{c}$ is not stable by rotation, one cannot use cyclicity to find several valid $\DiffSD$ equations from a unique one as in the previous case.
Therefore, to compute the probability that the $\DiffSD$ admits another solution than the one of its underlying $\QCSD$ instance when $\bm{c}$ is unstable by rotation, one has to consider the probability that all the $\Delta$ vectors $\bm{rot}_{\kappa_j}(\bm{y}_{\mu_j}) - \bm{c}$ have a preimage by $\bm{H}$ of weight $w$.
Each pre-image may exist with probability $p$ thus there exists $\Delta$ pre-images with probability $p^{\Delta} = \big({n \choose \omega}/2^{(n - k)}\big)^{\Delta}$.
As $2^{(n - k)} - 2$ possible values can be considered for $\bm{c}$ (all possible values except $\bm{0}$ and $\bm{1}$), it follows that $P[\mathtt{c_{unstable}}] = (2^{(n - k)} - 2) \times \big({n \choose \omega}/2^{(n - k)}\big)^{\Delta}$.

%Overall, this gives $\epsilon_{\QCSD} = (1 - M \times \frac{{n \choose \omega}}{2^{(n - k)}} - (2^{(n - k)} - 2) \times \big(\frac{{n \choose \omega}}{2^{(n - k)}}\big)^{\Delta}) \cdot \epsilon_{\DiffSD}$.

\vspace{\baselineskip}
\noindent \textbf{Theorem 3.}
\emph{If the hash function used is collision-resistant and if the commitment scheme used is binding and hiding, then the protocol depicted in Figure~\ref{fig:pok4-qc} is an honest-verifier zero-knowledge PoK for the $\QCSD(n, k, w, M)$ problem with soundness error equal to $\frac{1}{N} + \frac{(N - 1)(\Delta - 1)}{N \cdot M \cdot k}$ for some parameter $\Delta$.}

\begin{proof} 

The proofs of the correctness and special honest-verifier zero-knowledge properties follow the same arguments as the proofs given in Appendix~\ref{app:pok-sd1}.
Hereafter, we provide a proof for the $(Mk,N)$-special soundness property.

\vspace{\baselineskip}
  \noindent \textbf{$(Mk,N)$-special soundness.}
  To prove the $(Mk,N)$-special soundness, one need to build an efficient knowledge extractor $\Ext$ which returns a solution of the $\QCSD$ instance defined by $(\bm{H}, (\bm{y}_i)_{i \in \intoneto{M}})$ with high probability given a $(Mk,N)$-tree of accepting transcripts.
In our case, we build $\Ext$ as a knowledge extractor for the $\DiffSD$ problem and use it as extractor for the $\QCSD$ problem thanks to Theorem~\ref{thm:reduction-diffsd}.
One only need a subset of the tree of accepting transcripts to complete the proof namely $2\Delta$ leafs corresponding to challenges $\big( \mu_j, \kappa_j, \alpha_i \big)^{j \in \intoneto{\Delta}}_{i \in [1,2]}$.
The knowledge extractor $\Ext$ computes the solution as:

\vspace{0.5\baselineskip} 
\pseudocode{%
  \tx{1. Compute } (\bar{\pi}_i)_{i \in \intoneto{n}} \tx{ from } z_2^{(\mu_1, \kappa_1, \alpha_1)} \tx{ and } z_2^{(\mu_1, \kappa_1, \alpha_2)} \\
  \tx{2. Compute } \bm{c}_1 = \bm{H} \bm{z}^{(\mu_1, \kappa_1)}_1 - \rot_{\kappa_1}(\bm{y}_{\mu_1}) = \cdots = \bm{H} \bm{z}^{(\mu_{\Delta}, \kappa_{\Delta})}_1 - \rot_{\kappa_{\Delta}}(\bm{y}_{\mu_{\Delta}}) \\
  \tx{3. Compute } \bm{c}_2 = \bar{\pi}_{\alpha_1} \circ \cdots \circ \bar{\pi}_1[\bm{z}^{(\mu_1, \kappa_1)}_1] - \bm{z}^{(\mu_1, \kappa_1, \alpha_1)}_4 = \cdots = \bar{\pi}_{\alpha_1} \circ \cdots \circ \bar{\pi}_1[\bm{z}^{(\mu_{\Delta}, \kappa_{\Delta})}_1] \\ \hspace{68pt} - \bm{z}^{(\mu_{\Delta}, \kappa_{\Delta}, \alpha_{\Delta})}_4 \\
  \tx{4. Compute } \bm{c}_3 = \bm{H}(\bar{\pi}^{-1}_{1} \circ \cdots \circ \bar{\pi}^{-1}_{\alpha_1}[\bm{c}_2]) - \bm{c}_1 \\
  \tx{5. Compute } \bm{d}_j = \bar{\pi}^{-1}_{1} \circ \cdots \circ \bar{\pi}^{-1}_{\alpha_1}[\bm{z}^{(\mu_j, \kappa_j, \alpha_j)}_4] \tx{ for all } j \in \intoneto{\Delta} \\
  \tx{6. Output } (\bm{c}_3, (\bm{d}_{j}, \kappa_j, \allowbreak \mu_j)_{j \allowbreak \in \intoneto{\Delta}})
}
\vspace{0.5\baselineskip} 

\noindent One can compute $\big( \bar{\pi}^{(\mu_j, \kappa_j)}_i, \bar{\bm{v}}^{(\mu_j, \kappa_j)}_i \big)^{j \in \intoneto{\Delta}}_{i \in \intoneto{N}}$ from $\big( z_2^{(\mu_j, \kappa_j, \alpha_i)} \big)^{j \in \intoneto{\Delta}}_{i \in [1,2]}$.
  From the binding property of the commitments $(\com_{1,i})_{i \in \intoneto{N}}$, one can see that $(\pi_i, \bm{v}_i)_{i \in \intoneto{N}} \allowbreak = (\bar{\pi}^{(\mu_1, \kappa_1)}_i, \bar{\bm{v}}^{(\mu_1, \kappa_1)}_i)_{i \in \intoneto{N}} = \cdots = (\bar{\pi}^{(\mu_{\Delta}, \kappa_{\Delta})}_i, \bar{\bm{v}}^{(\mu_{\Delta}, \kappa_{\Delta})}_i)_{i \in \intoneto{N}}$.
From the binding property of commitment $\com_2$, one has $\bar{\bm{s}}^{(\mu_j, \kappa_j, \alpha_1)}_0 = \bar{\bm{s}}^{(\mu_j, \kappa_j, \alpha_2)}_0 = \bm{z}^{(\mu_j, \kappa_j)}_1$ for all $j \in \intoneto{\Delta}$.
In addition, one has $\bar{\bm{s}}^{(\mu_j, \kappa_j, \alpha_1)}_i = \bar{\pi}_i[\bar{\bm{s}}^{(\mu_j, \kappa_j, \alpha_1)}_{i - 1}] + \bar{\bm{v}}_i$ for all $i \in \intoneto{N} \setminus \alpha_1$ and all $j \in \intoneto{\Delta}$ as well as $\bar{\bm{s}}^{(\mu_j, \kappa_j, \alpha_2)}_i = \bar{\pi}_i[\bar{\bm{s}}^{(\mu_j, \kappa_j, \alpha_2)}_{i - 1}] + \bar{\bm{v}}_i$ for all $i \in \intoneto{N} \setminus \alpha_2$ and all $j \in \intoneto{\Delta}$.
  Using the binding property of commitment $\com_2$ once again, one can deduce that $\bar{\bm{s}}^{(\mu_j, \kappa_j)}_i = \bar{\pi}_i[\bar{\bm{s}}^{(\mu_j, \kappa_j)}_{i - 1}] + \bar{\bm{v}}_i$ for all $i \in \intoneto{N}$ and $j \in \intoneto{\Delta}$ hence $\bar{\bm{s}}^{(\mu_j, \kappa_j)}_{\alpha_1} = \bar{\pi}_{\alpha_1} \circ \cdots \circ \bar{\pi}_{1}[\bm{z}^{(\mu_j, \kappa_j)}_1] + \bm{v}_{\alpha_1} + \sum\nolimits_{i \in \intoneto{\alpha_1 - 1}} \bar{\pi}_{\alpha_1} \circ \cdots \circ \bar{\pi}_{i + 1}[\bm{v}_i]$ for all $j \in \intoneto{\Delta}$.
  From the binding property of commitment $\com_1$, one has $\bm{c}_1 = \bm{H} \bm{z}^{(\mu_1, \kappa_1)}_1 - \rot_{\kappa_1}(\bm{y}_{\mu_1}) = \cdots = \bm{H} \bm{z}^{(\mu_{\Delta}, \kappa_{\Delta})}_1 - \rot_{\kappa_{\Delta}}(\bm{y}_{\mu_{\Delta}})$.
  In addition, one has $\bar{\bm{r}}^{(\mu_1, \kappa_1)} = \cdots = \bar{\bm{r}}^{(\mu_{\Delta}, \kappa_{\Delta})}$ which implies that $\bar{\bm{t}}_{\alpha_1} = \bar{\bm{t}}^{(\mu_{\Delta}, \kappa_{\Delta})}_{\alpha_1} = \cdots = \bar{\bm{t}}^{(\mu_{\Delta}, \kappa_{\Delta})}_{\alpha_1}$.
From the binding property of commitment $\com_2$, one has $\bar{\bm{s}}^{(\mu_j, \kappa_j)}_{\alpha_1} = \bar{\bm{t}}^{(\mu_j, \kappa_j)}_{\alpha_1} + \bar{\bm{z}}^{(\mu_j, \kappa_j)}_{4}$ for all $j \in \intoneto{\Delta}$.
  Using $\bar{\bm{t}}_{\alpha_1} = \bar{\bm{t}}^{(\mu_1, \kappa_1)}_{\alpha_1} = \cdots = \bar{\bm{t}}^{(\mu_{\Delta}, \kappa_{\Delta})}_{\alpha_1}$, one can deduce that $\bm{c}_2 = \bar{\bm{t}}_{\alpha_1} - \bm{v}_{\alpha_1} - \sum\nolimits_{i \in \intoneto{\alpha_1 - 1}} \bar{\pi}_{\alpha_1} \circ \cdots \circ \bar{\pi}_{i + 1}[\bm{v}_i] = \bar{\pi}_{\alpha_1} \circ \cdots \circ \bar{\pi}_{1}[\bm{z}^{(\mu_1, \kappa_1)}_1] - \bm{z}^{(\mu_1, \kappa_1, \alpha_1)}_{4} = \cdots = \bar{\pi}_{\alpha_1} \circ \cdots \circ \bar{\pi}_{1}[\bm{z}^{(\mu_{\Delta}, \kappa_{\Delta})}_1] - \bm{z}^{(\mu_{\Delta}, \kappa_{\Delta}, \alpha_{\Delta})}_{4}$.
  It follows that $\bm{z}^{(\mu_j, \kappa_j)}_1 = \bar{\pi}^{-1}_{1} \circ \cdots \circ \bar{\pi}^{-1}_{\alpha_1}[\bm{c}_2 + \bm{z}^{(\mu_j, \kappa_j, \alpha_j)}_{4}]$ for all $j \in \intoneto{\Delta}$.
  As $\bm{c}_1 = \bm{H} \bm{z}^{(\mu_j, \kappa_j)}_1 - \rot_{\kappa_j}(\bm{y}_{\mu_j})$, one has $\bm{c}_1 = \bm{H}(\bar{\pi}^{-1}_{1} \circ \cdots \circ \bar{\pi}^{-1}_{\alpha_1}[\bm{c}_2 + \bm{z}^{(\mu_j, \kappa_j, \alpha_j)}_{4}]) - \rot_{\kappa_j}(\bm{y}_{\mu_j})$ hence $\bm{H}(\bar{\pi}^{-1}_{1} \circ \cdots \circ \bar{\pi}^{-1}_{\alpha_1}[\bm{z}^{(\mu_j, \kappa_j, \alpha_j)}_{4}]) + \bm{c}_3 = \rot_{\kappa_j}(\bm{y}_{\mu_j})$ for all $j \in \intoneto{\Delta}$.
  Given that $\hw{\bm{z}^{(\mu_j, \kappa_j, \alpha_j)}_4} = \omega$ for all $j \in \intoneto{\Delta}$, one can conclude that $(\bm{c}_3, (\bm{d}_{j}, \kappa_j, \allowbreak \mu_j)_{j \in \intoneto{\Delta}})$ is a solution of the considered $\DiffSD$ problem instance.
One completes the proof by using Theorem~\ref{thm:reduction-diffsd}.
\end{proof}

\section{Proof of Theorem~\ref{thm:pok-rsd2}} \label{app:pok-rsd2}

\noindent \textbf{Theorem 4.}
\emph{If the hash function used is collision-resistant and if the commitment scheme used is binding and hiding, then the protocol depicted in Figure~\ref{fig:pok-rsd2} is an honest-verifier zero-knowledge PoK for the $\IRSL$ problem with soundness error equal to $\frac{1}{N} + \frac{(N - 1)(\Delta - 1)}{N (q^{Mk} - 1)}$ for some parameter $\Delta$.}

\begin{proof} 

  The proof of our protocol in the rank metric (Theorem \ref{thm:pok-rsd2}) is similar to the proof of our protocol in Hamming metric (Theorem \ref{thm:pok-sd2-qc}) presented in Appendix~\ref{app:pok-sd2-qc}.
It relies on the introduction of the intermediary $\DiffIRSL$ problem (Definition \ref{def:diffrsl}) along with a reduction from the $\IRSL$ problem to the $\DiffIRSL$ problem (Theorem~\ref{thm:reduction-diffrsl}).

\end{proof}

\begin{definition}[$\DiffIRSL$ problem] \label{def:diffrsl}
  Let $(q, m, n = 2k, k, w, M, \Delta)$ be positive integers, $P \in \Fq[X]$ be an irreducible polynomial of degree $k$, $\bm{H} \in \mathcal{ID}(\Fqm^{\nmktn})$ be a random parity-check matrix of an ideal code of index $2$, $E$ be a random subspace of $\Fqm$ of dimension $\omega$, $(\bm{x}_i)_{i \in \intoneto{M}} \in (\Fqm^n)^M$ be random vectors such that $\rsup{\bm{x}_i} = E$ and $(\bm{y}_i)_{i \in \intoneto{M}} \in (\Fqm^{(\nmk)})^M$ be vectors such that $\bm{H} \bm{x}_i^\top = \bm{y}_i^\top$.
  Given $(\bm{H}, (\bm{y}_i)_{i \in \intoneto{M}})$, the Differential Ideal Rank Support Learning problem $\IRSL(q, m, n, k, \allowbreak w, M, \Delta)$ asks to find $(\bm{c}, (\bm{d}_{\delta}, (\gamma_{i,j}^{\delta})_{i \in \intoneto{M}, j \in \intoneto{k}})_{\delta \in \intoneto{\Delta}}) \in \Ft^{(n - k)} \times (\Ft^{n} \times (\Fq^{Mk} \setminus (0, \cdots, 0))^{\Delta})$ such that $\bm{H} \bm{d}_{\delta}^\top + \bm{c} = \sum\nolimits_{(i, j) \in \intoneto{M} \times \intoneto{k}} \gamma^{\delta}_{i, j} \cdot \bm{rot}_{j}(\bm{y}_{i}^{\top})$ with $\rsup{\bm{d}_{\delta}} = F$ and $|F| = \omega$ for each $\delta \in \intoneto{\Delta}$.
\end{definition}

\begin{theorem} \label{thm:reduction-diffrsl}
  If there exists a $\ppt$ algorithm solving the $\DiffIRSL(q, m, n, k, \allowbreak w, M, \Delta)$ with probability $\epsilon_{\DiffIRSL}$, then there exists a $\ppt$ algorithm solving the $\IRSL(q, m, n, \allowbreak k, w, M)$ problem with probability $\epsilon_{\IRSL} \ge (1 - (q^{m(n - k)} - 1) \times \allowbreak (q^{\omega (m - \omega) + n\omega - m(n - k)})^{\Delta}) \cdot \epsilon_{\DiffIRSL}$.
\end{theorem}

\noindent \textit{Sketch of proof.}
The proof of Theorem \ref{thm:reduction-diffrsl} in the rank metric setting is similar to the proof of Theorem \ref{thm:reduction-diffsd} in the Hamming metric setting.
A noticeable difference in the rank metric setting is related to the use of the irreducible polynomial $P \in \Fq[X]$ of degree $k$.
Indeed, the later implies that $P[\mathtt{c_{stable}}] = 0$ namely there is no solution where both $\bm{c}$ is stable by rotation and $\bm{c} \neq (0, \cdots, 0)$.

Given an $[n, k]$ ideal code $\mathcal{C}$, we restrict our analysis to the case where the weight $\omega$ is lower than the rank Gilber-Varshamov bound associated to $\mathcal{C}$ \textit{i.e.} the value for which the number of words of weight less or equal $w$ corresponds to the number of syndromes.
As a consequence, given a syndrome $\bm{y}$, the probability that there exists a pre-image $\bm{x}$ of $\bm{y}$ such that $\bm{Hx}^\top = \bm{y}^\top$ and $\rw{\bm{x}} = \omega$ is $q^{\omega(m - \omega) + n \omega - m(n - k)}$ where $q^{\omega(m - \omega)}$ is an approximation of the Gaussian binomial which counts the number of vector spaces of dimension $\omega$ in $\Fqm$, $q^{n\omega}$ is the number of words in a basis of dimension $\omega$ and $q^{m(n - k)}$ is the number of syndromes.
As such, this probability describes the number of codewords of rank weight $\omega$ divided by the number of possible syndromes.
Following the same steps than the proof of Theorem~\ref{thm:reduction-diffsd} (with $\sum\nolimits_{(i, j) \in \intoneto{M} \times \intoneto{k}} \gamma^{\delta}_{i, j} \cdot \bm{rot}_{j}(\bm{y}_{i}^{\top}) - \bm{c}$ playing the role of $\bm{rot}_{\kappa_j}(\bm{y}_{\mu_j}) - \bm{c}$) and taking into account that $P[\mathtt{c_{stable}}] = 0$, one get $\epsilon_{\IRSL} \ge (1 - (q^{m(n - k)} - 1) \times (q^{\omega (m - \omega) + n\omega - m(n - k)})^{\Delta}) \cdot \epsilon_{\DiffIRSL}$.

\section{PoK leveraging structure related to the $\MQ$ problem} \label{app:extra1}

\begin{center}
  \resizebox{1\textwidth}{!}{
    \pseudocode{%
      \hspace{440pt} \\[-2\baselineskip][\hline]\\[-8pt]
      \underline{\pcalgostyle{Inputs~\&~Public~Data}} \\
      w = (\bm{x}_i)_{i \in \intoneto M}, ~ x = (\mqF, (\bm{y}_i)_{i \in \intoneto{M}}), ~
      \chsps = \intoneto{M} \times \Fq^*, ~ \chs = (\mu, \kappa) \\[\baselineskip]
      \underline{\prover_1(w, x)} \\
      \theta \sampler \bit^{\lambda} \\
      \pcfor i \in \intoneto{N - 1} \pcdo \\
      \pcind \theta_{i} \samples{\theta} \bit^{\lambda}, ~ \phi_{i} \samples{\theta_i} \bit^{\lambda},
      ~ \bm{u}_i \samples{\phi_i} \Fq^n, ~ \bm{v}_i \samples{\phi_i} \Fq^m,
      ~ r_{1, i} \samples{\theta_i} \bit^{\lambda}, ~ \com_{1, i} = \commith{r_{1,i}, \, \phi_i} \\
      \pcend \\
      \theta_{N} \samples{\theta} \bit^{\lambda}, ~ \phi_{N} \samples{\theta_N} \bit^{\lambda}, ~ \bm{u}_N \samples{\phi_N} \Fq^n,
      ~ \bm{u}  =  \sum\nolimits_{i \in \intoneto{N}} \bm{u}_i, ~ \bm{v}_N = \mqF(\bm{u}) - \sum\nolimits_{i \in \intoneto{N - 1}} \bm{v}_i \\
      r_{1, N} \samples{\theta_N} \bit^{\lambda}, ~ \com_{1, N} = \commith{r_{1,N}, \, \bm{v}_N \, || \, \phi_N} \\
      \com_1 = \commitb{(\com_{1,i})_{i \in \intoneto{N}}} \\[\baselineskip]
      \underline{\prover_2(w, x, \mu, \kappa)} \\
      \bm{s}_0 = \kappa \cdot \bm{x}_{\mu} - \bm{u} \\
      \pcfor i \in [1, N] \pcdo \\
      \pcind \bm{s}_i = \mqG(\bm{u}_i, \bm{s}_0) + \bm{v}_i \\
      \pcend \\
      \com_2 = \commitb{\bm{s}_0 \, || \, (\bm{s}_i)_{i \in \intoneto{N}}} \\[\baselineskip]
      \underline{\prover_3(w, x, \mu, \kappa, \alpha)} \\
      \bm{z}_1 = \kappa \cdot \bm{x}_{\mu} - \bm{u} \\
      \pcif \alpha \neq N \pcdo \\
      \pcind z_2 = \bm{v}_N \, || \, (\theta_{i})_{i \in \intoneto{N} \setminus \alpha} \\
      \pcelse \\
      \pcind z_2 = (\theta_{i})_{i \in \intoneto{N} \setminus \alpha} \\
      \pcend \\
      \rsp = (\bm{z}_1, z_2, \com_{1,\alpha}) \\[\baselineskip]
      \underline{\verifier(x, \com_1, (\mu, \kappa), \com_2, \alpha, \rsp)} \\
      \tx{Compute } (\bar{\phi_i}, \bar{r}_{1,i}, \bar{\bm{u}}_i, \bar{\bm{v}}_i)_{i \in \intoneto{N} \setminus \alpha} \tx{ from } z_2 \\
      \bar{\bm{s}}_0 = \bm{z}_1 \\
      \pcfor i \in [1, N] \setminus \alpha \pcdo \\
      \pcind \bar{\bm{s}}_i = \mqG(\bar{\bm{u}}_i, \bar{\bm{s}}_0) + \bar{\bm{v}}_i \\
      \pcend \\
      \bar{\bm{s}}_{\alpha} = \kappa^2 \cdot \bm{y}_{\mu} - \mqF(\bm{z}_1) - \sum\nolimits_{i \in \intoneto{N} \setminus \alpha} \bar{\bm{s}}_i, ~ \bar{\com}_{1, \alpha} = \com_{1, \alpha} \\
      \pcfor i \in [1, N] \setminus \alpha \pcdo \\
      \pcind \pcif i \neq N \pcdo \\
      \pcind \pcind \bar{\com}_{1,i} = \commith{\bar{r}_{1,i}, \, \bar{\phi}_i} \\
      \pcind \pcelse \\
      \pcind \pcind \bar{\com}_{1,N} = \commith{\bar{r}_{1,N}, \, \bar{\bm{v}}_N \, || \, \bar{\phi}_N} \\
      \pcind \pcend \\
      \pcend \\
      b_1 \samplen \big( \com_{1} = \commitb{(\com_{1,i})_{i \in \intoneto{N}}} \big), ~
      b_2 \samplen \big( \com_2 = \commitb{\bm{z}_1 \, || \, (\bar{\bm{s}}_i)_{i \in \intoneto{N}}} \big) \\ 
      \pcreturn b_1 \wedge b_2 
      ~\\[3pt][\hline]\\[-2\baselineskip]
    }
  }
  \captionof{figure}{PoK leveraging structure for the $\HMQP$ problem \label{fig:mq}}
\end{center}

\section{PoK leveraging structure related to $\SD$ over $\Fq$} \label{app:extra2}

  \begin{center}
    \resizebox{1\textwidth}{!}{
      \pseudocode{%
        \hspace{440pt} \\[-2\baselineskip][\hline]\\[-8pt]
        \underline{\pcalgostyle{Inputs~\&~Public~Data}} \\
        w = \bm{x}, ~ x = (\bm{H}, \bm{y}) \\
        \chsps = \Fq^*, ~ \chs = \kappa \\[\baselineskip]
        \underline{\prover_1(w, x)} \\
        \theta \sampler \bit^{\lambda}, ~ \xi \sampler \bit^{\lambda} \\
        \pcfor i \in \intoneto{N} \pcdo \\
        \pcind \theta_{i} \samples{\theta} \bit^{\lambda}, ~ \phi_{i} \samples{\theta_i} \bit^{\lambda},
        ~, \pi_i \samples{\phi_i} \hperm{n}, ~ \bm{v}_i \samples{\phi_i} \Fq^n,
        ~, r_{1, i} \samples{\theta_i} \bit^{\lambda}, ~ \com_{1, i} = \commith{r_{1,i}, \, \phi_{i}} \\
        \pcend \\
        \pi = \pi_N \circ \cdots \circ \pi_1,
        ~ \bm{v} = \bm{v}_N + \sum\nolimits_{i \in \intoneto{N - 1}} \pi_N \circ \cdots \circ \pi_{i + 1}[\bm{v}_i],
        ~ \bm{r} \samples{\xi} \Fq^n, ~ \bm{u} = \pi^{-1}[\bm{r} - \bm{v}] \\
        \com_1 = \commitb{\bm{H} \bm{u} \, || \, \pi[\bm{u}]+ \bm{v} \, || \, \pi[\bm{x}] \, || \, (\com_{1,i})_{i \in \intoneto{N}}} \\[\baselineskip] 
        \underline{\prover_2(w, x, \kappa)} \\
        \bm{s}_0 = \bm{u} + \kappa \cdot \bm{x} \\
        \pcfor i \in [1, N] \pcdo \\
        \pcind \bm{s}_i = \pi_i[\bm{s}_{i - 1}] + \bm{v}_i \\ 
        \pcend \\
        \com_2 = \commitb{\bm{u} + \kappa \cdot \bm{x} \, || \, (\bm{s}_i)_{i \in \intoneto{N}}} \\[\baselineskip]
        \underline{\prover_3(w, x, \kappa, \alpha)} \\
        \bm{z}_1 = \bm{u} + \kappa \cdot \bm{x}, ~ z_2 = (\theta_{i})_{i \in \intoneto{N} \setminus \alpha},
        ~ z_3 = \xi, ~ \bm{z}_4 = \pi[\bm{x}] \\
        \rsp = (\bm{z}_1, z_2, z_3, \bm{z}_4, \com_{1,\alpha}) \\[\baselineskip]
        \underline{\verifier(x, \com_1, \kappa, \com_2, \alpha, \rsp)} \\
        \tx{Compute } (\bar{\phi_i}, \bar{r}_{1,i}, \bar{\pi}_i, \bar{\bm{v}}_i)_{i \in \intoneto{N} \setminus \alpha} \tx{ from } z_2
        \tx{ and } \bar{\bm{r}} \tx{ from } z_3 \\
        \bar{\bm{t}}_N = \bar{\bm{r}}, ~ \bar{\bm{b}}_N = \bm{z}_4 \\
        \pcfor i \in \{N, \ldots, \alpha + 1\} \pcdo \\
        \pcind \bar{\bm{t}}_{i-1} = \bar{\pi}^{-1}_i[\bar{\bm{t}}_i - \bar{\bm{v}}_i], ~ \bar{\bm{b}}_{i-1} = \bar{\pi}^{-1}_i[\bar{\bm{b}}_i] \\
        \pcend \\
        \bar{\bm{s}}_0 = \bm{z}_1, ~ \bar{\bm{s}}_{\alpha} = \bar{\bm{t}}_\alpha + \kappa \cdot \bar{\bm{b}}_{\alpha}, ~ \bar{\com}_{1, \alpha} = \com_{1, \alpha} \\
        \pcfor i \in [1, N] \setminus \alpha \pcdo \\
        \pcind \bar{\bm{s}}_i = \bar{\pi}_i[\bar{\bm{s}}_{i - 1}] + \bar{\bm{v}}_i, 
        ~ \bar{\com}_{1,i} = \commith{\bar{r}_{1,i}, \, \bar{\phi_i}} \\
        \pcend \\
        b_1 \samplen \big( \com_{1} = \commitb{\bm{H} \bm{z}_1 - \kappa \cdot \bm{y} \, || \, \bar{\bm{r}} \, || \, \bm{z}_4 \, || \, (\bar{\com}_{1,i})_{i \in \intoneto{N}}} \big) \\
        b_2 \samplen \big( \com_2 = \commitb{\bm{z}_1 \, || \, (\bar{\bm{s}}_i)_{i \in \intoneto{N}}} \big) \\ 
        b_3 \samplen \big( \hw{\bm{z}_4} = \omega \big) \\
        \pcreturn b_1 \wedge b_2 \wedge b_3
      ~\\[3pt][\hline]\\[-\baselineskip]
      }
    }
    \captionof{figure}{PoK leveraging structure for the $\SD$ problem over $\Fq$ \label{fig:pok4-fq}}
  \end{center}
  \vspace{\baselineskip}

%\section{\textcolor{red}{Proof of Theorem~\ref{thm:pok-mq2}}} \label{app:pok-mq2}
%\input{proof-hmq}

%\section{PoK leveraging structure for the $\SD$ problem over $\Fq$} \label{app:pok-sd2-fq}
%\input{pok-sd-fq}

%\iflong
%\else
%\input{appendix-short}
%\fi
%--------------------------------------------------------------------%

\end{document}